\documentclass[11pt]{article}
\usepackage{graphicx} 
\usepackage[margin=1in]{geometry} 
\usepackage{amsmath,amsthm,amssymb, mathtools}

\DeclareMathOperator*{\argmin}{arg\,min}
\usepackage{nicematrix}
\usepackage{verbatim}
\usepackage[normalem]{ulem}
\usepackage{algorithm}
\usepackage{algpseudocode}
\usepackage{authblk}
\usepackage[T1]{fontenc}
\usepackage{arydshln}

\usepackage{thmtools}
\usepackage{thm-restate}

\newcommand{\R}{\mathbb{R}}

\newcommand{\knote}[1]{{\bf{\color{blue}[\tiny Karthik: #1]}}}
\newcommand{\cnote}[1]{{\bf{\color{red}[\tiny Chengyue: #1]}}}
\newcommand{\ynote}[1]{{\bf{\color{orange}[\tiny Yuri: #1]}}}

\newtheorem{theorem}{Theorem}

\newtheorem{claim}[theorem]{Claim}
\newtheorem{lemma}[theorem]{Lemma}
\newtheorem{corollary}[theorem]{Corollary}

\newtheorem{definition}[theorem]{Definition}
\newtheorem{example}[theorem]{Example}

\usepackage{color-edits}

\usepackage{hyperref}
\hypersetup{colorlinks=true,breaklinks=true,bookmarks=true,urlcolor=blue,citecolor=blue,linkcolor=blue,bookmarksopen=false,draft=false}

\title{Scarf's Algorithm on Arborescence Hypergraphs}
\author{Karthekeyan Chandrasekaran\thanks{University of Illinois, Urbana-Champaign, karthe@illinois.edu}}
\author{Yuri Faenza\thanks{Columbia University, yuri.faenza@columbia.edu}}
\author{Chengyue He\thanks{Columbia University, ch3480@columbia.edu}}
\author{Jay Sethuraman\thanks{Columbia University, jay@ieor.columbia.edu}}
\affil{} 
\date{}

\begin{document}

\maketitle


\begin{abstract}
Scarf’s algorithm---a pivoting procedure that finds a dominating extreme point in a down-monotone polytope---can be used to show the existence of a fractional stable matching in hypergraphs. The problem of finding a fractional stable matching in a hypergraph, however, is PPAD-complete. In this work, we study the behavior of Scarf's algorithm on arborescence hypergraphs, the family of hypergraphs in which hyperedges correspond to the paths of an arborescence. For arborescence hypergraphs, we prove that Scarf's algorithm can be implemented to find an integral stable matching in polynomial time. En route to our result, we uncover novel structural properties of bases and pivots for the more general family of network hypergraphs. Our work provides the first proof of polynomial-time convergence of Scarf's algorithm on hypergraphic stable matching problems, giving hope to the possibility of polynomial-time convergence of Scarf's algorithm for other families of polytopes.
\end{abstract}

\section{Introduction}
Scarf~\cite{scarf1967core} proved the existence of a core allocation in a large class of cooperative games with non-transferable utility.
Key ingredients in this proof include a lemma---{\em Scarf's lemma}---that asserts the existence of a dominating extreme point in certain polytopes, and a pivoting procedure---{\em Scarf's algorithm}---to find one.
Scarf's results have profoundly influenced subsequent research in combinatorics, theoretical computer science, and game theory: Scarf's lemma has been used to show the existence of fair allocations such as cores and fractional cores~\cite{biro2016fractional,scarf1967core}, strong fractional kernels~\cite{aharoni1995fractional}, fractional stable solutions in hypergraphs~\cite{aharoni2003lemma}, and stable paths~\cite{haxell2008fractional}. Yet, it is not known how to efficiently construct these desirable allocations/solutions: It is widely believed that Scarf's algorithm is unlikely to be efficient for arbitrary applications.  
In this work, we investigate the possibility of efficient convergence of Scarf's algorithm for finding a stable matching in hypergraphs.

The stable matching problem in bipartite graphs (also known as the stable marriage problem) is a classic and well-studied problem. It is well-known that a stable matching in a bipartite graph with preferences always exists~\cite{gale1962college} and can be found in polynomial time via multiple algorithms, including Gale and Shapley's deferred acceptance algorithm~\cite{gale1962college}, linear programming~\cite{baiou2000stable,roth1993stable,rothblum1992characterization,teo1998geometry} and other combinatorial algorithms~\cite{baiou2002stable,dworczak2016deferred}. Recently, it was shown that Scarf's algorithm can be implemented to converge in polynomial time for the stable marriage problem \cite{faenza2023scarf}. This result motivated us to address the possibility of efficient convergence of Scarf's algorithm for hypergraph stable matching problems. 

In contrast to graphs, the problem of finding a stable matching in hypergraphs is significantly more challenging.  Even in the special case of tripartite $3$-regular hypergraphs, it is NP-complete to find a stable matching~\cite{ng1991three} (this is also known as the \emph{stable family problem} proposed by Knuth~\cite{knuth1976marriages}). Despite the hardness results of finding (integral) stable matchings in hypergraphs, Aharoni and Fleiner~\cite{aharoni2003lemma} used Scarf's lemma to show that a fractional stable matching exists in every hypergraph. 
However, this general implication from Scarf's result comes at a computational price: the problem of 
computing a fractional stable matching on hypergraphs is PPAD-complete~\cite{kintali2013reducibility}. The latter problem remains PPAD-complete even in extremely restrictive cases, for example, even if the hypergraph has node degree~\cite{ishizuka2018complexity} or hyperedge size~\cite{csaji2022complexity} upper bounded by $3$.

In this paper, we define a new class of hypergraphs called {\em arborescence} hypergraphs (in which the hyperedges are paths in an arborescence), and design a polynomial-time algorithm to find an integer stable matching. We accomplish this through a pivoting rule that implements Scarf's algorithm and terminates in a polynomial number of iterations. The node-hyperedge incidence matrix of an arborescence hypergraph is totally unimodular and consequently, every extreme point of the associated fractional matching polytope is integral. Thus, every dominating extreme point is also integral. This observation and Aharoni and Fleiner's results imply that the extreme point found by Scarf's algorithm for arborescence hypergraphs corresponds to an \emph{integral} stable matching. Our pivoting rule can be repurposed to apply to the more general family of
{\em network hypergraphs}\footnote{The node-hyperedge incidence matrix of such a hypergraph is totally unimodular and consequently there is always an integer stable matching.} but we are not able to prove polynomial convergence for this more general family.

\paragraph{Our Contributions.}
A hypergraph  $H=(V,E)$ is specified by a vertex set $V$ and a hyperedge set $E$, where every hyperedge $e\in E$ is a subset of $V$.
A \emph{hypergraphic preference system} is given by a pair $(H, \succ)$, where $H=(V,E)$ is a hypergraph and $\succ:=\{\succ_i:i\in V\}$ is the preference profile, $\succ_i$ being a strict order over $\delta(i)=\{e\in E:i\in e \}$ for each $i\in V$. For $e, e' \in \delta(i)$, we write $e\succeq_i e'$ if either $e\succ_i e'$ or $e=e'$. In addition, we assume that for every $i\in V$, the \emph{singleton hyperedge} $e_i=\{i\}\in\delta(i)$ and for every $e'\in\delta(i)$, $e'\succeq_i e_i$\footnote{This assumption corresponds, in the bipartite setting, to the usual hypothesis that an agent prefers to be matched rather than being unmatched.}. 



A \emph{stable matching} for a hypergraphic preference system is a vector $x\in\{0,1\}^E$ so that for every $e\in E$, there exists a vertex $i\in e$ such that
\begin{equation}\label{eq:sm}
    \sum_{e'\in\delta(i),e'\succeq_i e}x_{e'}=1.
\end{equation}
Equation~\eqref{eq:sm} with $e=e_i$ imposes that $x$ is the characteristic vector of a matching. Moreover, for every hyperedge $e$, there exists a vertex $i \in e$ and a matching edge $e'$ so that $e'\succeq_i e$. Correspondingly, if a fractional vector $x\in[0,1]^E$ still satisfies~\eqref{eq:sm}, then we call such a fractional vector $x$ a \emph{fractional stable matching}.

An \emph{arborescence} is a directed graph that contains a vertex $r$ (called the \emph{root}) such that every vertex $v\neq r$ has a unique directed path from $r$. A hypergraph $H=(V,E)$ is an \emph{arborescence hypergraph} if there exists an arborescence $\mathcal{T}=(U, \mathcal{A}_0)$ such that $V=\mathcal{A}_0$ and each hyperedge $e\in E$ is a subset of arcs in $\mathcal{A}_0$ that forms a directed path. 
We note that there exists a polynomial-time algorithm to verify whether a given hypergraph is an arborescence hypergraph and moreover, find an associated arborescence $\mathcal{T}$~\cite{schrijver1998theory}. 

As our main result, we show that Scarf's algorithm can be implemented to run in polynomial time for every arborescence hypergraphic preference system.

\begin{restatable}{theorem}{main}
\label{thm:main-arb-poly}
Let $(H=(V,E),\succ)$ be a hypergraphic preference system where $H$ is an arborescence hypergraph. There exists a pivoting rule such that Scarf's algorithm terminates in at most $|V|$ iterations and outputs a stable matching on $(H=(V,E),\succ)$ in time $O(|V||E|)$.
\end{restatable}
A few remarks on this result are in order. To the best of our knowledge, this is the first result showing polynomiality of Scarf's algorithm beyond the case of stable marriage~\cite{faenza2023scarf}. We note that the pivoting rule used to show polynomiality of Scarf's algorithm in stable marriage~\cite{faenza2023scarf} was heavily inspired by Gale and Shapley's classical deferred acceptance mechanism, which is a purely combinatorial algorithm. In contrast, the pivoting rule developed in this work does not seem to have a combinatorial counterpart and is substantively different from the one from~\cite{faenza2023scarf}. In Section~\ref{sec:negative}, we provide evidence suggesting that classical approaches to stable marriage problems, such as a natural linear program with an exact description of the convex hull of all stable matchings~\cite{teo1998geometry}, do not extend to even a subclass of arborescence hypergraphic preference systems.

Secondly, while most of the known results related to polynomiality of stable matching on hypergraphic preference system usually restricts either the degree of nodes~\cite{ishizuka2018complexity} or the size of hyperedges~\cite{csaji2022complexity}, we do not assume any condition on them or on the preference lists of agents. Thirdly, previous work by~\cite{csaji2022complexity} shows that there is a polynomial time algorithm for finding a stable matching on arborescence hypergraphic systems. They reduce the problem of finding a stable matching on such instances to finding a \emph{kernel} in the clique-acyclic superorientation of a chordal graph, which can be solved in polynomial time \cite{PIM20}. 
We remark that the algorithm from~\cite{PIM20} is different from Scarf's algorithm and does not seem to generalize to network hypergraphs, while, in our investigation, we uncover novel properties of bases and pivots in network hypergraphs, which may be of independent interest and could lead to polynomial-time convergence proofs for this more general class of hypergraphs.

We recall standard concepts on Scarf's algorithm, implication of Scarf's lemma to the stable matching problem in hypergraphs, and formally define arborescence hypergraphs in Section \ref{sec:prelims}. We prove Theorem \ref{thm:main-arb-poly} in Section \ref{sec:Poly-Scarf-Cases}. In Section~\ref{sec:negative}, we show that a natural relaxation of the convex hull of stable matchings in an arborescence hypergraph is fractional.

\section{Preliminaries}\label{sec:prelims}
We recall some standard concepts in this section as preliminaries. In Section~\ref{sec:pre-Scarf}, we present Scarf's lemma. In Section~\ref{sec:pre-Scarf-SM}, we explain its implications to the stable matching problem in hypergraphs. In Section \ref{sec:pre-Scarf-alg}, we give details about Scarf's algorithm---in particular we describe cardinal pivot and ordinal pivot operations and how they are combined to show the existence of a dominating basis. In Section~\ref{sec:pre-network}, we define network matrices and network hypergraphs. In Section~\ref{sec:pre-arborescence}, we define arborescence hypergraphs, the special case of interest to this work. 

Throughout the paper, we use the following standard notations: For $n\in\mathbb{Z}_+$, we let $[n]:=\{1,2,\dots,n\}$. For sets $S,S'$ with $|S\setminus S'|=1$, we abuse notation and let $S\setminus S'$ be the unique element $e\in S\setminus S'$. For a matrix $A\in\mathbb{R}^{n\times m}$, $i\in[n], j\in[m]$, we let $a_{i,j}$ be the entry in the $i$-th row and $j$-th column of $A$. We denote the vector corresponding to $j$-th column by $A_j$. For $B\subseteq[m]$, we denote the submatrix of $A$ corresponding to the columns in $B$ as $A_B$.

\subsection{Scarf's Lemma}\label{sec:pre-Scarf}
Scarf's lemma deals with the standard form of down-monotone polytopes, as defined below.

\begin{definition}[Standard form, cardinal basis, extreme point]\label{def:dm-polytope}
    Let $n\le m$ be positive integers. Let $A\in\mathbb{R}^{n\times m}_{\ge 0}$ be a nonnegative matrix. Matrix $A$ is in \emph{standard form} if $A$ has the form $A=(I_n|A')$ where $I_n$ is the $n\times n$ identity matrix. Let $b\in\mathbb{R}^n_+$ be strictly positive. Consider the polytope

\begin{equation}\label{eq:Ax-leq-b}
P=\{x \in \mathbb{R}^{m}_{\ge 0} : Ax = b\}.
\end{equation} 
A set $B\subseteq[m]$ is a \emph{cardinal basis} for $(A,b)$ if $|B|=n$ and the submatrix $A_B$ of $A$ indexed by the columns in $B$ has full row rank. In addition, let $B$ be a cardinal basis and 
$x=(x_B,0)\in\mathbb{R}^m$ be such that $\bar{x}=x_B$ is the unique solution to the linear system $A_B\bar{x}=b$. Then, if $x\in P$, we call $B$ a \emph{feasible cardinal basis}, and the solution $x$ is an \emph{extreme point} of $P$ corresponding to $B$.
\end{definition}

In addition to the polytope description, Scarf's lemma also takes as input a matrix $C$. Since no entry appears twice in any row of $C$, this matrix can be interpreted as a row-wise ordering of columns.

\begin{definition}[Ordinal matrix, ordinal basis, utility vector]\label{def:ordinal-matrix}


Let $n\le m$ be positive integers. Let $C=(c_{i,j}) \in \R^{n\times m}$ be a matrix. Matrix $C$ is an \emph{ordinal matrix} if (i)~for every $i\in[n]$ and $j\neq k\in[m]$, $c_{i,j}\neq c_{i,k}$, and (ii)~for every distinct $i,j \in [n]$, $k \in [m]\setminus [n]$, $c_{i,i} < c_{i,k} < c_{i,j}$. Let $O$ be a set of columns of $C$. For every row $i \in [n]$, define the utility of the row (w.r.t. $O$) as
\begin{equation}\label{eq:Utility-Vector}
    u_i^O:=\min_{j\in O}c_{i,j}.
\end{equation}
The set $O$ is called an \emph{ordinal basis} of $C$ if $|O|=n$ and for every column $j\in[m]$, there is at least one row $i \in [n]$ such that $u^O_i\ge c_{i,j}$. The associated vector $u^O\in\mathbb{R}^n$ is called the \emph{utility vector} of the ordinal basis $O$.
\end{definition}

For two real vectors $x,y$ in the same dimension, we succinctly write $x>y$ if $x_i>y_i$ for every entry $i$, and write $x\ngtr y$ if $x>y$ fails. We remark that, for every set $O$ of $n$ columns, we have that $C_j\ge u^O$ for every $j\in O$ by definition. Hence, a set $O$ of $n$ columns is an ordinal basis iff for every column $j\in[m]$, the column vector $C_j$ satisfies $C_j\ngtr u^O$.

\begin{definition}[Dominating basis]
    A feasible cardinal basis for $(A,b)$ that is also an ordinal basis of $C$ is called a \emph{dominating basis} for $(A,b,C)$. Given a dominating basis $B$ for $(A,b,C)$, 
    the extreme point $x$ of $P=\{x \in \mathbb{R}^{m}_{\ge 0} : Ax = b\}$ corresponding to $B$ is called a \emph{dominating extreme point} for $(A,b,C)$.
\end{definition}

Scarf's lemma, presented next, shows the existence of a dominating extreme point under mild conditions on $(A, b, C)$.

\begin{theorem}[Scarf's Lemma]\label{thm:scarflemma}
Let $A\in \R^{n \times m}_{\ge 0}$ be a matrix in standard form and $b\in\R^n_{+}$ such that the 
$P=\{x\in \R^m_{\ge 0}: Ax=b\}$ is bounded. 
Let $C \in \R^{n \times m}$ be an ordinal matrix. Then, there exists a dominating extreme point for $(A,b,C)$. 
\end{theorem}

\subsection{Implication of Scarf's Lemma for Fractional Hypergraph Stable Matching}\label{sec:pre-Scarf-SM}

Aharoni and Fleiner~\cite{aharoni2003lemma} used Scarf's Lemma to establish the existence of a fractional stable matching in hypergraphs. In this section, we briefly explain their approach via a special case.

\begin{definition}\label{def:block}
    Let $(H=(V,E),\succ)$ be a hypergraphic preference system with $e_i:=\{i\}$ for $i\in [n]$ being the singleton hyperedges. We say that matrices $A,C$ are \emph{block-partitioned} if they are constructed as follows (see Example \ref{example:hypergraph-matrices} for an illustration): 
    \begin{enumerate}
        \item For each $i\in V$, we define the $i$-th block $S_i:=\{e\in E\setminus\{e_i\}:i=\max_{j\in e} j\}$. 
        
        \item Let $A$ be the $V\times E$ incidence matrix of $H$ whose columns are ordered as follows: The columns corresponding to the singleton hyperedges are the first $n$ columns of $A$. Next, we group the remaining hyperedges into $n$ blocks $S_1, \ldots, S_n$. The $i$-th block consists of $S_i$ and is ordered in decreasing order of preference from left to right according to $\succ_{i}$. The $i$-th block appears before the $i+1$-th block for every $i\in [n-1]$. 
        
        \item The ordinal matrix $C$ has the same number of rows and columns as $A$, and each column index $j\in [m]$ corresponds to the same hyperedge $e_j$ in both $A$ and $C$. We now describe the entries in $C$. For each $i\in[n],j\in[m]$ with $a_{i,j}=1$, we set $c_{i,j}=|\delta(i)|-\ell$ if $e_j$ is the $\ell$-th best hyperedge with respect to $\succ_i$. The remaining entries in row $i$ are assigned integers that are no less than $|\delta(i)|$ and in decreasing order from left to right. 
    \end{enumerate}
\end{definition}
Let $(H=(V,E),\succ)$ be a hypergraphic preference system. Let $S_1, \ldots, S_n$ be the blocks and $A, C$ be the block-partitioned matrices associated with this hypergraphic preference system. We observe that $S_1$ is empty since the only hyperedge $e$ such that $\max_{j\in e} j=1$ is $e=\{1\}$, which is a singleton. Aharoni and Fleiner \cite{aharoni2003lemma} showed that every dominating extreme point for $(A,b,C)$ is a fractional stable matching of the hypergraphic preference system $(H=(V,E),\succ)$. 

\begin{theorem}[\cite{aharoni2003lemma}]\label{thm:Domin-SM}
Let $(H=(V,E),\succ)$ be a hypergraphic preference system. 
Let $A$ and $C$ be the block-partitioned matrices associated with this hypergraphic preference system and let $b\in {1}^V$ be the all ones vector. Then, the matrix $A$ is in standard form, the vector $b$ is a positive vector, the polytope $P=\{x\in \R^E_{\ge 0}: Ax=b\}$ is bounded, and the matrix $C$ is ordinal. Moreover, every dominating extreme point for $(A,b,C)$ is a fractional stable matching of $(H,\succ)$.

\end{theorem}

We note that not every fractional stable matching is a dominating extreme point for $(A,b,C)$ (see~\cite{faenza2023scarf}). By Theorem~\ref{thm:scarflemma} and Theorem~\ref{thm:Domin-SM}, every hypergraphic preference system admits a fractional stable matching.

\begin{example}\label{example:hypergraph-matrices}
    For $V=\{1,2,3,4\}$, consider the hypergraph in Figure~\ref{fig:Hypergraph_Ex} and the following preference list:
    \begin{align*}
        1 &: \{1,3\}\succ_1 \{1,3,4\}\succ_1 \{1\}, \\
        2 &: \{2,3\}\succ_2 \{2\}, \\
        3 &: \{2,3\}\succ_3 \{1,3\}\ \succ_3 \{3,4\}\succ_3 \{1,3,4\} \succ_3 \{3\}, \\
        4 &: \{1,3,4\}\succ_4 \{3,4\}\succ_4 \{4\}.
    \end{align*}
    \begin{figure}[h!]
    \centering
    \includegraphics[scale=1]{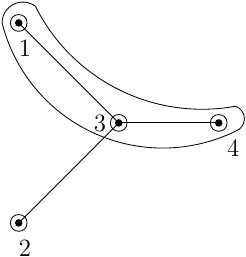}
    \caption{An example of a hypergraph. The circles around the nodes are singletons. Line segments represent edges. Other hyperedges (only $\{1,3,4\}$ in this example) are indicated by splinegons. }
    \label{fig:Hypergraph_Ex}
\end{figure}

    The blocks are $S_1=\emptyset, S_2=\emptyset, S_3=\{\{2,3\},\{1,3\}\},S_4=\{\{1,3,4\},\{3,4\}\}$. The block-partitioned incidence matrix and ordinal matrix are as follows:
    \begin{displaymath}
    A=\begin{pNiceMatrix}[first-row,first-col]
     & \{1\} & \{2\} & \{3\} & \{4\} & \{2,3\} & \{1,3\} & \{1,3,4\} & \{3,4\} \\
    1 & 1 & 0 & 0 & 0 & 0 & 1 & 1 & 0 \\
    2 & 0 & 1 & 0 & 0 & 1 & 0 & 0 & 0 \\
    3 & 0 & 0 & 1 & 0 & 1 & 1 & 1 & 1 \\
    4 & 0 & 0 & 0 & 1 & 0 & 0 & 1 & 1  
    \end{pNiceMatrix},
    \end{displaymath}

\begin{displaymath}
        C=\begin{pNiceMatrix}[first-row,first-col]
     & \{1\} & \{2\} & \{3\} & \{4\} & \{2,3\} & \{1,3\} & \{1,3,4\} & \{3,4\} \\
    1 & 0 & 7 & 6 & 5 & 4 & 2 & 1 & 3 \\
    2 & 7 & 0 & 6 & 5 & 1 & 4 & 3 & 2 \\
    3 & 7 & 6 & 0 & 5 & 4 & 3 & 1 & 2 \\
    4 & 7 & 6 & 5 & 0 & 4 & 3 & 2 & 1 
\end{pNiceMatrix}.
    \end{displaymath}

\end{example}

\subsection{Scarf's Algorithm}\label{sec:pre-Scarf-alg}

We review Scarf's algorithm in this section. It involves the iterative repetion of two operations, namely cardinal pivot and ordinal pivot. We describe these two operations in Sections \ref{sec:pre-cp} and \ref{sec:pre-op} and then show how they are combined in Scarf's algorithm in Section \ref{sec:iteration}. 

\subsubsection{Cardinal Pivot}\label{sec:pre-cp}
We define the cardinal pivot operation in this section. 


\begin{definition}[Cardinal pivot]\label{def:cpdef}
Let $A\in \R^{n\times m}$ be a matrix in standard form and $b\in \R^n$ be a vector. 
Let $B=\{j_1,\dots,j_n\}$ be a feasible cardinal basis and let $j_t\in[m]\setminus B$. A \emph{cardinal pivot} from $B$ with \emph{entering column} $j_t$ returns a feasible cardinal basis $B'$ that is defined as follows:
Let $x$ be the extreme point of $P$ corresponding to $B$. Let $A_B,A_{j_t}$ be the submatrices of $A$ indicated by $B,j_t$, respectively. Let $y=A_B^{-1}A_{j_t}\in\mathbb{R}^B$. A column $j\in B$ is a \emph{leaving candidate} if $y_j>0$ and $j\in\argmin_{j\in B:y_j>0}\{x_j/y_j\}$. We define $B':=B\cup\{j_t\}\setminus\{j_\ell\}$, where $j_\ell$ is a leaving candidate. The column $j_\ell=B-B'$ is known as the leaving column.
\end{definition}
The following is a standard result in linear programming, which implies that the cardinal pivot is well-defined. 
\begin{lemma}[\cite{scarf1967core}]\label{lem:scarf-cp}
    If $P=\{x\in\mathbb{R}^m_{\ge 0}:Ax=b\}$ is bounded, then for every feasible cardinal basis $B$ and every $j_t\in [m]\setminus B$, there exists at least one leaving candidate. Also, if $j_\ell$ is a leaving candidate, then $B'=B\cup\{j_t\}\setminus\{j_\ell\}$ is a feasible cardinal basis.
\end{lemma}


We note that a cardinal pivot could have multiple leaving candidates. A pivoting rule determines a unique leaving candidate. 

\begin{definition}[Pivoting rule, Degeneracy]\label{def:degenerate}
Let $A\in \R^{n\times m}$ be a matrix in standard form and $b\in \R^n$ be a vector. 
A \emph{pivoting rule} is a criterion that determines a unique leaving candidate for every tuple $(B, j_t)$, where $B$ is a feasible cardinal basis and $j_t\in [m]\setminus B$.
A cardinal pivot $B\to B'$ is \emph{degenerate} if the corresponding extreme points $x$ and $x'$ are such that $x=x'$. Otherwise, it is \emph{non-degenerate}.
\end{definition}

\subsubsection{Ordinal Pivot}\label{sec:pre-op}
We define the ordinal pivot operation in this section. In order to define the ordinal pivot operation, we need the following definition. 

\begin{definition}[Disliked relation]
    Let $C$ be an ordinal matrix, $O$ be an ordinal basis, and let $i\in[n]$. A column $j\in O$ is \emph{$i$-disliked} w.r.t.~$O$ if $u_i^O=c_{ij}$. If $O$ is clear from context, then we say that $j$ is $i$-disliked (and omit ``w.r.t.~$O$'').
\end{definition}

The ``disliked'' relation gives a bijection between the row set $[n]$ and the column set $O$ as shown in the following lemma.
\begin{lemma}\label{lem:disliked-oneone}
    Let $C$ be an ordinal matrix and $O$ be an ordinal basis. Then, for every $j\in O$, there is a unique $i\in[n]$ such that column $j$ is $i$-disliked. Also, for every distinct $j,j'\in O$, if $j$ is $i$-disliked and $j'$ is $i'$-disliked, then $i\neq i'$.
\end{lemma}

\begin{proof}
    By definition, since $O$ is an ordinal basis, we have $C_j\ngtr u^O$ for every $j\in[m]$. Suppose that there exists $j\in O$ such that no row dislikes $j$. Then for every row $i\in[n]$, $u^O_i\neq c_{i,j}$ implies $u^O_i<c_{i,j}$, which results in $u^O<C_j$, a contradiction. Therefore, every $j\in O$ is disliked by at least one row. Since $|O|=n$, we deduce that every $j\in O$ is disliked by exactly one row $i\in [n]$. 

    Moreover, by definition of $u^O$ and the fact that every row vector in $C$ has distinct entries, no row can simultaneously dislike two columns. Therefore, the second part of the lemma holds.
\end{proof}

We now have the ingredients needed to define the ordinal pivot. Pseudocode is provided in Algorithm~\ref{alg:op}.

\begin{definition}[Ordinal pivot]\label{def:appopdef}
Let $C$ be an ordinal matrix, $O$ be an ordinal basis, and $j_\ell\in O$. An \emph{ordinal pivot} from $O$ with leaving column $j_\ell$ returns $O'=(O\setminus \{j_\ell\}) \cup \{j^*\}$ where $j^*$ is defined as follows: 
Let $i_{\ell}$ be the unique row such that $j_{\ell}$ is $i_{\ell}$-disliked w.r.t.~$O$ (exists and unique by Lemma \ref{lem:disliked-oneone}). 
Let $j_r$ be the column such that $u^{O-j_\ell}_{i_\ell}=c_{i_\ell,j_r}$ (recall that $u^{O-j_\ell}\in\mathbb{R}^n$ is the utility vector w.r.t. $O\setminus\{j_\ell\}$). 
Let $i_r$ be the unique row in $[n]$ such that $j_r$ is $i_r$-disliked w.r.t.~$O$ (exists and unique by Lemma \ref{lem:disliked-oneone}). Define
\begin{align}\label{eq:appchoiceoford}
    K&:=\{k\in[m]\setminus O:c_{i,k}>u^{O-j_\ell}_i, \textrm{ for all } i\neq i_r\},\\
    j^*&:=\mathop{\arg\max}_{k\in K} c_{i_r, k}. \notag
\end{align}
We call $i_r$ as the \emph{reference row} and $j_r$ as the \emph{reference column}. 
\end{definition}

\begin{algorithm}
\caption{Ordinal Pivot}\label{alg:op}
\begin{algorithmic}
\State{Let $O$ be the current ordinal basis, associated with utility vector $u^O$. Suppose that $j_\ell\in O$ is the leaving column.}
\State{$i_\ell\gets \textrm{the unique row $i\in[n]$ that makes $u_i^O\le c_{i,j_\ell}$ equal.}$}
\State{$j_r\gets  \mathop{\arg\min}_{j\in O-j_\ell}\{c_{i_\ell,j}\}$.}\Comment{Reference column.}
\State{$i_r\gets\textrm{the unique row $i\in[n]$ that makes $u_i^O\le c_{i,j_r}$ equal.}$}
\State{$K\gets\{k\in[m]\setminus O:c_{i,k}>u^{O-j_\ell}_i, \textrm{ for all } i\neq i_r\}$.}
\State{$j^*\gets\mathop{\arg\max}_{k\in K} c_{i_r, k}$.}\Comment{Entering column.}
\State{$O\gets O\cup\{j^*\}\setminus\{j_\ell\}$.}
\end{algorithmic}
\end{algorithm}

In contrast to the cardinal pivot, the ordinal pivot reverses the order of entering and leaving. Moreover, the ordinal pivot does not need a pivoting rule since the operation is unique. Scarf showed the following result which implies that the ordinal pivot is well-defined. 
\begin{lemma}[\cite{scarf1967core}]\label{lem:op-def}
    Let $C$ be an ordinal matrix, $O$ be an ordinal basis, and $j_\ell\in O$. Let $i_{\ell}$, $i_r$, $j_r$, and $j^*$ be as defined in Definition \ref{def:appopdef}. Then, we have the following: 
\begin{enumerate}
    \item Column $j=i_r\in K$, and thus $K\neq\emptyset$.
    \item The set $O'=(O\setminus\{j_\ell\})\cup\{j^*\}$ is an ordinal basis.
    \item Let $j\in[m]\setminus O$. If $O''=(O\setminus\{j_\ell\})\cup\{j\}$ is an ordinal basis, then $j=j^*$. In other words, with a fixed leaving column $j_\ell$, there exists a unique ordinal basis $O''$ with $|O\cap O''|=n-1$ and $j_{\ell}\not\in O'$. 
\end{enumerate}
\end{lemma}

We recall that every ordinal basis is associated with an utility vector. Scarf showed certain helpful properties about the utility vectors: the row $i_\ell$ that dislikes the leaving column $j_\ell$ in the ordinal basis $O$ will increase its utility, while the reference row $i_r$ that dislikes the reference column $j_r$ will decrease its utility. The utilities of the other rows stay the same. We state his result below since it will be useful in designing a potential function for poly-time convergence of Scarf's algorithm in certain settings. 

\begin{lemma}[\cite{scarf1967core}]\label{lem:utility}
    Let $C$ be an ordinal matrix, $O$ be an ordinal basis, and $j_\ell\in O$. Let $i_{\ell}$, $i_r$, $j_r$, $j^*$, and $O'$ be as defined in Definition \ref{def:appopdef}. 
    Then, 
    \begin{enumerate}
    \item $u^{O'}_{i_\ell}=c_{i_\ell,j_r}>c_{i_\ell,j_\ell}=u^O_{i_\ell}$, 
    \item $u^{O'}_{i_r}=c_{i_r,j^*}<c_{i_r,j_r}=u^O_{i_r}$, and 
    \item $u^{O'}_{i}=u^O_{i}$ for every $i \in [n]\setminus\{i_{\ell}, i_r\}$. 
    \end{enumerate}
\end{lemma}

\subsubsection{Initialization, Iteration, and Termination}\label{sec:iteration}
Let $A\in \R^{n\times m}$ be a matrix in standard form, $b\in \R^m_{+}$ be a positive vector such that $P=\{x\in \R^m_{\ge 0}: Ax = b\}$ is bounded, and $C$ be an ordinal matrix. Fix a cardinal pivoting rule. Scarf's algorithm initializes with a basis pair $(B_0,O_0)$, where $B_0=\{1,2,\dots,n\}$ and $O_0=\{2,3,\dots,n,n+1\}$. 
The algorithm is iterative and maintains a pair $(B,O)$ where $B$ is a feasible cardinal basis, $O$ is an ordinal basis and $|B\cap O|\ge n-1$ (it can be verified that the initialized pair $(B_0, O_0)$ satisfies these conditions). It proceeds alternatively with a cardinal pivot using the cardinal pivoting rule and an ordinal pivot. In particular, it goes through a sequence
\begin{equation}\label{eq:iterations}
    (B_0,O_0)\to (B_1,O_0)\to (B_1,O_1)\to \cdots,
\end{equation}
where for every $i\ge 0$, we have that $(B_i,O_i)\to (B_{i+1},O_i)$ is a cardinal pivot from the feasible cardinal basis $B_i$ with entering column $j_t=O_i-B_i$ using the cardinal pivoting rule, and $(B_{i+1},O_i)\to(B_{i+1},O_{i+1})$ is an ordinal pivot from the ordinal basis $O_i$ with leaving column $j_\ell=B_{i+1}-O_i$.

\begin{definition}[Scarf pair, iteration]\label{def:Scarf-pair}
    A pair $(B_i,O_i)$ in the sequence~\eqref{eq:iterations} with $B_i\neq O_i$ is called a \emph{Scarf pair}. We denote a consecutive pair of cardinal pivot and ordinal pivot $(B_i,O_i)\to(B_{i+1},O_i)\to(B_{i+1},O_{i+1})$ as an \emph{iteration} of Scarf's algorithm.
\end{definition}

We notice that $|B_0\cap O_0|=n-1$. The algorithm will maintain the invariant $|B_i\cap O_i|\ge n-1$ and $|B_{i+1}\cap O_i|\ge n-1$. The sequence terminates when 
$|B\cap O|=n$. 
If the sequence terminates, then the algorithm returns the extreme point $x$ corresponding to the feasible cardinal basis $B$. 
Therefore, an iteration of Scarf's algorithm begins with a Scarf pair and returns either another Scarf pair or a dominating basis. Scarf showed the following result:

\begin{theorem}\cite{scarf1967core}\label{thm:constructive-scarf}
Let $A\in \R^{n\times m}$ be a matrix in standard form, $b\in \R^m_{+}$ be a positive vector such that $P=\{x\in \R^m_{\ge 0}: Ax = b\}$ is bounded, and $C$ be an ordinal matrix. There exists a cardinal pivoting rule such that Scarf's algorithm executed with that cardinal pivoting rule for $(A, b, C)$ terminates within finite number of iterations. Moreover, the extreme point $x$ associated with the terminating feasible cardinal basis $B$ is a dominating extreme point. 
\end{theorem}

Theorem \ref{thm:constructive-scarf} implies Scarf's lemma (Theorem \ref{thm:scarflemma}). 
The algorithm also provides a procedure to find a fractional stable matching for a given hypergraphic preference stystem (via  Theorem~\ref{thm:Domin-SM}). However, (similar to the simplex algorithm) the convergence may not be efficient for arbitrary $(A, b, C)$ satisfying the hypothesis of Theorem \ref{thm:constructive-scarf}. 

\subsection{Network Matrix and Network Hypergraph}\label{sec:pre-network}

Network hypergraphs are a special family of hypergraphs for which a hypergraphic preference system admits a stable matching (not just a  fractional stable matching). In this section, we introduce network hypergraphs and certain subfamilies of network hypergraphs that are of interest to this work. We first recall certain terminology. A directed graph $\mathcal{D}=(U,\mathcal{A})$, is specified by a finite set $U$ of \emph{vertices} and a set $\mathcal{A}=\{(u,v):u,v\in U\}$ of \emph{arcs}.
\begin{definition}[$\mathcal{D}$-Path, forward/backward arcs, and $\mathcal{D}$-directed path]\label{def:Mono-Path}
Let $\mathcal{D}=(U,\mathcal{A})$ be a directed graph. Let $F\subseteq \mathcal{A}$, and $P=(a_1,a_2,\dots,a_p)$ be an ordered set such that $F=\{a_1,a_2,\dots,a_p\}$.
\begin{enumerate}
    \item We say $F$ (or $P$) is a \emph{$\mathcal{D}$-path} if $P$ forms an undirected path. When the underlying digraph is clear from context, we omit $\mathcal{D}$ and say $F$ (or $P$) is a path.
    \item An arc $a_i\in F$ is \emph{forward} on $P$ if following the order $P$ we visit the tail of $a_i$ before visiting its head, otherwise $a_i$ is \emph{backward} on $P$.
    \item We say $F$ (or $P$) is a \emph{$\mathcal{D}$-directed path} if every arc in $P$ is a forward arc on $P$.
\end{enumerate}
\end{definition}

\begin{definition}[Path concatenation]\label{def:walk}
Let $P=(a_1,\dots,a_p)$ and $P'=(a_1',\dots,a_{p'}')$ be two paths. The starting vertex (resp.~ending vertex) of $P$ is the unique vertex of $a_1\setminus a_2$ (resp.~$a_p\setminus a_{p-1}$). If the ending vertex of $P$ is the same as the starting vertex of $P'$, we define $P\oplus P'=(a_1,\dots,a_p,a'_1,\dots,a'_{p'})$ as the \emph{concatenation} of the two paths.
\end{definition}

The concatenation of $P$ and $P'$ may not be a path since $P$ and $P'$ may share common vertices.

\begin{definition}[Directed tree]\label{def:direct-tree}
    A directed graph $\mathcal{T}$ is a \emph{directed tree} if the underlying undirected graph (obtained by dropping the orientation of all arcs) is connected and acyclic. 
\end{definition}



\begin{definition}[Network Matrix]\label{def:network-matrix}
Let $\mathcal{D}=(U,\mathcal{A})$ be a directed graph and $\mathcal{T}=(U,\mathcal{A}_0)$ be a directed tree on vertex set $U$. The network matrix corresponding to $(\mathcal{D}, \mathcal{T})$ is the matrix $M\in\{0,\pm 1\}^{\mathcal{A}_0\times \mathcal{A}}$ where for every $a=(u,v)\in\mathcal{A}$ and $a'\in \mathcal{A}_0$, we have 
\begin{equation}\label{eq:Network-Matrix}
    M_{a',a}=\left\{
    \begin{array}{cc}
        1 & \textrm{if $a'$ is a forward arc on the unique $\mathcal{T}$-path from $u$ to $v$;}\\
        -1 & \textrm{if $a'$ is a backward arc on the unique $\mathcal{T}$-path from $u$ to $v$;} \\
        0 & \textrm{if $a'$ does not belong to the unique $\mathcal{T}$-path from $u$ to $v$.}
    \end{array}
    \right.
\end{equation}

A matrix $M$ is a \emph{network matrix} if there exists a directed graph $\mathcal{D}=(U,\mathcal{A})$ and a directed tree $\mathcal{T}=(U,\mathcal{A}_0)$ such that the network matrix corresponding to $(\mathcal{D}, \mathcal{T})$ is $M$. We say $\mathcal{T}$ is the \emph{principal tree} corresponding to $M$.
\end{definition}
We note that a network matrix $M$ corresponding to $(\mathcal{D}=(U,\mathcal{A}),\mathcal{T}=(U,\mathcal{A}_0))$ is in standard form if and only if $\mathcal{A}_0\subseteq \mathcal{A}$ and the first $n$ columns of $M$ correspond to $\mathcal{A}_0$. 
Network matrices are totally unimodular (see,~e.g.~\cite{schrijver1998theory}). A consequence of this fact is the following: 

\begin{theorem}\label{thm:Network-Matrix-TU}
    Let $M\in\{0,\pm 1\}^{n\times m}$ be a network matrix. Then, all extreme points of the polyhedron $\{x\in\mathbb{R}^m_{\ge 0}:Mx=b\}$ with $b\in\mathbb{Z}^n$ are integer vectors.
\end{theorem}

We will focus on $\{0,1\}$-valued network matrices. A $\{0,1\}$-valued matrix $M$ can naturally be represented by a hypergraph whose node-hyperedge incidence matrix is $M$. If $M$ is additionally a network matrix, then we call the associated hypergraph as a network hypergraph. We define them formally below. See Figure \ref{fig:Network_HG} for an illustration.
\begin{definition}[Network Hypergraph]\label{def:Network_HG}
    A hypergraph $H=(V,E)$ is a \emph{network hypergraph} if the node-hyperedge incidence matrix $A$ of $H$ is a network matrix. We say $(\mathcal{D},\mathcal{T})$ is the \emph{underlying network} of $H$ if the node-hyperedge incidence matrix $A$ of $H$ is the network matrix corresponding to $(\mathcal{D},\mathcal{T})$.
\end{definition}

\begin{figure}[h!]
    \centering
    \includegraphics[scale=1.05]{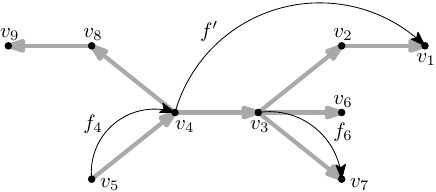}
    \hspace{1cm}
    \includegraphics[scale=1.05]{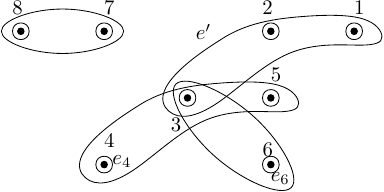}
    \caption{An example of a network hypergraph. On the left we present the principal tree using grey arcs. On the right we have a hypergraph with $V=[8]$ and 12 hyperedges (8 of them are singletons). We illustrate the hyperedges $e_4=\{4\}$, $e_6={6}$ and $e'=\{1,2,3\}$ on the right by the paths $f_4$, $f_6$ and $f'$ using black arcs on the left, respectively. In the principal tree we present, the unique source is $v_5$. The sinks are $v_1,v_6,v_7,v_9$. There are two branching vertices, namely $v_3$ and $v_4$.
    }
    \label{fig:Network_HG}
\end{figure}

We note that there exists a polynomial time algorithm to verify whether a given matrix is a network matrix and if so, then find the principal tree~\cite{schrijver1998theory}. 
Therefore, we can also use the same algorithm to verify whether a given hypergraph is a network hypergraph (equivalent to verifying whether a given $\{0,1\}$-valued matrix is a network matrix). By Theorems~\ref{thm:scarflemma},~\ref{thm:Domin-SM} and~\ref{thm:Network-Matrix-TU}, we have the following corollary:

\begin{corollary}\label{cor:Exist-SM}
    Let $H=(V, E)$ be a network hypergraph. For every preference set $\succ$, the hypergraphic preference system $(H=(V,E),\succ)$ admits a stable matching.
\end{corollary}

\subsection{Arborescence Hypergraph, Interval Hypergraph}\label{sec:pre-arborescence}

We focus primarily on the network hypergraphic preference system where the principal tree $\mathcal{T}$ of the underlying network $(\mathcal{D},\mathcal{T})$ is an \emph{arborescence}.

\begin{definition}[Arborescence]\label{def:Arb}
    A directed tree $\mathcal{T}=(U,\mathcal{A}_0)$ is an \emph{arborescence} if there is a distinguished vertex $r\in U$ called \emph{root} such that for every $v\in U$, the unique $\mathcal{T}$-path from $r$ to $v$ is a $\mathcal{T}$-directed path.
\end{definition}

\begin{definition}[Arborescence Hypergraph]\label{def:Arb-HG}
    A network hypergraph $H=(V,E)$ with underlying network $(\mathcal{D}=(U,\mathcal{A}),\mathcal{T}=(U,\mathcal{A}_0))$ is an \emph{arborescence hypergraph} if the principal tree $\mathcal{T}=(U,\mathcal{A}_0)$ is an arborescence. 
\end{definition}


\begin{figure}[h!]
    \centering
    \includegraphics[scale=1.05]{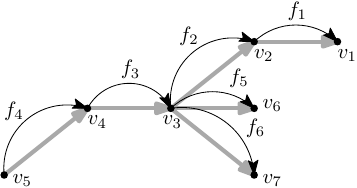}
    \hspace{1.5cm}
    \includegraphics[scale=1.05]{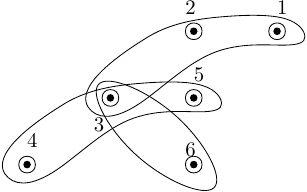}
    \caption{An example of an arborescence and an arborescence hypergraph. On the left we present an arborescence $\mathcal{T}$ with the root $v_5$. On the right we have a hypergraph $H$ with principal tree $\mathcal{T}$ with $V=[6]$ and 9 hyperedges (6 of them are singletons). We present the arcs in the arborescence which are corresponding to the singletons in the hypergraph.}
    \label{fig:Arborescence_HG}
\end{figure}


\begin{definition}[Interval Hypergraph]
    A hypergraph $H=(V,E)$ is an \emph{interval hypergraph} if $V=[n]$ and every $e\in E$ is such that $e=\{i, i+1, \ldots,j\}$ for some $i,j\in[n]$ with $i\le j$.
\end{definition}

An interval hypergraph is an arborescence hypergraph: Indeed, given an interval hypergraph $H=(V,E)$ with $V=[n]$, let $\mathcal{T}=(U,\mathcal{A}_0)$ be a directed graph where $U=[n+1]$ and $\mathcal{A}_0=\{(i,i+1):i\in[n]\}$. Let $\mathcal{D}=(U,\mathcal{A})$ be such that $\mathcal{A}=\{(i,j+1):i\le j,[i,j]\in E\}$. Then $H$ is a network hypergraph with underlying network $(\mathcal{D},\mathcal{T})$ where $\mathcal{T}$ is an arborescence.


\begin{figure}[h!]
    \centering
    \includegraphics[scale=1.1]{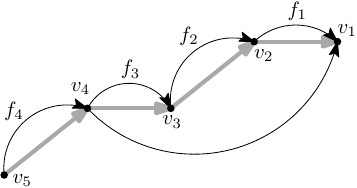}
    \hspace{1.5cm}
    \includegraphics[scale=1.1]{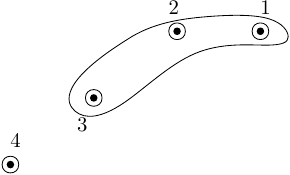}
    \caption{An example of an interval hypergraph. On the left we present the principal tree. On the right we have an interval hypergraph with $V=[4]$. The hyperedge $\{1,2,3\}$ can be seen as an interval $[1,3]$, and correspond to the black arc/grey path from $v_4$ to $v_1$ on the principal tree.}
    \label{fig:Interval_HG}
\end{figure}


\begin{definition}[source, sink, branching vertex, leaf]\label{def:Source-Sink-Branching}
    Let $\mathcal{T}=(U,\mathcal{A}_0)$ be a directed tree on $U$ and let $v\in U$. We denote $\delta^+_{\mathcal{T}}(v)=\{(v,u):u\in U,(u,v)\in\mathcal{A}_0\}$ and $\delta^-_{\mathcal{T}}(v)=\{(w,v):w\in U,(w,v)\in\mathcal{A}_0\}$. Let $d^+_{\mathcal{T}}(v):=|\delta^+_{\mathcal{T}}(v)|$ and $d^-_{\mathcal{T}}(v):=|\delta^-_{\mathcal{T}}(v)|$, and we call them out-degree, in-degree of $v$ in $\mathcal{T}$, respectively. If $\delta^+_{\mathcal{T}}(v)=0$ (resp.,~$\delta^-_{\mathcal{T}}(v)=0$), then we call $v$ a \emph{sink} (resp.,~\emph{source}). If $\delta^+_{\mathcal{T}}(v)=2$, then we say that $v$ is a \emph{branching vertex}. An arc $e\in\mathcal{A}_0$ is called a \emph{leaf} if its tail vertex is a sink.
\end{definition}

See Figure~\ref{fig:Network_HG} for an illustration of source, sink, branching vertex, and leaf. 
It can be verified that an arborescence is a directed tree with a unique source. 
A natural related question is the following: is it possible to verify whether a given network hypergraph is an arborescence hypergraph in polynomial-time? We leave this as an open question. 
For the purposes of this paper, we assume that the input arborescence hypergraph is given by its underlying network where the principal tree is an arborescence.

\section{Polynomiality of Scarf's Algorithm on Arborescence Hypergraphic Preference System}\label{sec:Poly-Scarf-Cases}

According to Corollary~\ref{cor:Exist-SM}, every hypergraphic preference system over an arborescence hypergraph admits a stable matching. The key contribution of this paper is showing that a stable matching in an arborescence hypergraphic preference system can be found in polynomial time via Scarf's algorithm.




\main*

To show Theorem~\ref{thm:main-arb-poly}, we establish certain properties of Scarf's algorithm applied to a hypergraphic preference system. From now on, let $(H=(V,E),\succ)$ be the hypergraphic preference system of interest. Let $A$ and $C$ be the block-partitioned matrices associated with this hypergraphic preference system and let $b\in \{1\}^V$ be the all-ones vector. 
In Section~\ref{sec:scarf-deep}, we discuss some general properties of Scarf pairs arising in the execution of Scarf's algorithm on hypergraphic preference system; in particular, we define the notion of a \emph{separator} associated with each Scarf pair which will act as a potential function to bound the number of iterations. In Section~\ref{sec:Pivot-NM}, we focus on the case where $H$ is a network hypergraph. For this family, we give a combinatorial interpretation of the cardinal pivot. In Section~\ref{sec:DFO} and Section~\ref{sec:FFL}, we focus on the case where $H$ is an arborescence hypergraph. We design a cardinal pivoting rule. In Section~\ref{sec:run-time}, we prove Theorem~\ref{thm:main-arb-poly} based on the properties established in Sections \ref{sec:scarf-deep}--\ref{sec:FFL}.


\subsection{Ordinal Pivots for Hypergraphic Preference System}\label{sec:scarf-deep}

Let $(H=(V,E),\succ)$ be the hypergraphic preference system of interest. Let $A$ and $C$ be the block-partitioned matrices associated with this hypergraphic preference system and let $b\in \{1\}^V$ be the all-ones vector. 

\subsubsection{Controlling node} 
We first make an assumption about $(A,b,C)$. When implementing Scarf's algorithm, the index $i=1$ can be seen as a controlling row/column in Scarf's algorithm (see~\cite{scarf1967core}). Without loss of generality, we can let the first row and column be artificial, which eliminates the asymmetry among original rows/agents. Formally, we have the following:

\begin{lemma}\label{lem:0-Row-Column}
Let $(A,b,C)$ be as above. Construct a new tuple $(A',b',C')$ where 
\begin{equation}\label{eq:0-Row-Column}
    A'=\left(\begin{matrix}
        1 & 0^T \\
        0 & A
    \end{matrix}\right),
    b'=\left(\begin{matrix}
        1  \\
        b
    \end{matrix}\right),
    C'=\left(\begin{matrix}
        0 & \xi^T \\
        \chi & C
    \end{matrix}\right),
\end{equation}
where $\chi=(M,\dots,M)^T$ such that $M>\max_{i,j}c_{ij}$, and $\xi=(m,m-1,\dots,2,1)^T$. We index the new row/column in $A',C'$ by $0$-th row/column.
Then $B$ is a dominating basis for $(A,b,C)$ if and only if $B'=B\cup\{0\}$ is a dominating basis for $(A',b',C')$. 

\end{lemma}

\begin{proof}
For an $n \times n$ submatrix $A_B$ of $A$, we have that 
$$
A_{B'}=\begin{pNiceArray}{ccccc}[first-row,first-col]
     & 0 & &\cdots & \cdots & \cdots \\
     0 & 1 &\Vdots & 0 & \cdots & 0\\ \hdottedline
     \vdots & 0 & & \Block{3-3}{A_B} \\
     \vdots & \vdots & & \\
     \vdots & 0 & & \\
\end{pNiceArray}
$$
is an $(n+1)\times (n+1)$ submatrix of $A'$, with $\det{A_{B'}}\neq 0$ if and only if $\det{A_B}\neq 0$. Also, for $x\in\R^n_{\ge 0}$, $x'=(1,x^T)^T$ yields that $A_Bx=b$ if and only if $A_{B'}x'=b'$. Thus equivalently, $B'$ is an $(A',b')$-feasible basis if and only if $B$ is an $(A,b)$-feasible basis.\par
Similarly, for an $n\times n$ submatrix ${C_O}$ of $C$ we define the following $(n+1)\times (n+1)$ submatrix of $C'$ as
$$
C_O'=\begin{pNiceArray}{ccccc}[first-row,first-col]
     & 0 & &\cdots & \cdots & \cdots \\
     0 & 0 &\Vdots & \times & \cdots & \times\\ \hdottedline
     \vdots & M & & \Block{3-3}{C_O} \\
     \vdots & \vdots & & \\
     \vdots & M & & \\
\end{pNiceArray}.
$$

By definition, if $O$ is an ordinal basis of $C$, then the utility vector $u^O$ cannot have $u^O<C_j$ for any column $j\in [n+m]$. Thus, the utility vector $u^{O'}=(0,(u^O)^T)^T$ generated by $O'$ cannot have $u^{O'}<C'_j$ for any column $j\in\{0,1,\dots,n+m\}$. Thus, $O'$ is an ordinal basis for $C'$. Conversely, if $O'$ is an ordinal basis of $C'$, then again by definition, utility vector $u^{O'}$ has the form $u^{O'}=(0,(u^O)^T)^T$, where $u^O_i=\min_{\ell\in O}c_{i,\ell}$ for $i\in [n]$, since by assumption we have $c_{i,\ell}<M$. By the dominating property $u^{O'}<C_j$ is impossible for every column $j\in\{0,1,\dots,n+m\}$. Since $u^{O'}_0=0<c_{0j}$, we must have $u^O_i=u^{O'}_i\ge c_{i,j}$ for some $i\in[n]$, which implies the dominating property for $O$. Hence, $B$ is a dominating basis for $(A,b,C)$ if and only if $B'$ is a dominating basis for $(A',b',C')$.
\end{proof}

Another nice property of this modified instance is that the terminating pivot of Scarf's algorithm is always an ordinal pivot.

\begin{lemma}\label{lem:last-ordinal}
    Let $(A',b',C')$ be as in Lemma~\ref{lem:0-Row-Column}, let $(B_0,O_0)$ be the initial Scarf pair and consider the sequence of pivot~\eqref{eq:iterations}. Then, for every $i\ge 0$, $B_{i+1}\neq O_i$. Therefore, Scarf's algorithm terminates when $B_i=O_i$ for some $i\in\mathbb{Z}_+$, which happens only after an ordinal pivot.
\end{lemma}

\begin{proof}
    Let $(B,O)$ be a Scarf pair. Then we have $|B\cap O|=n$ (notice that $|B|=|O|=n+1$ since we have $0$-th row) and $0=B-O$. Consider the cardinal pivot $(B,O)\to (B',O)$, we claim that $0\in B'$ and thus $B'\neq O$. In fact, if $0\notin B'$, then the $0$-th row of the submatrix $A_{B'}$ has all $0$'s, which implies $A_{B'}$ is not full-rank and thus $B'$ is not a cardinal basis, a contradiction. Therefore, since $0\in B'$ and $0\notin O$, we have $B'\neq O$. Thus, the algorithm does not terminate and perform the next step, which is an ordinal pivot $(B',O)\to (B',O')$. Since a cardinal pivot can not make the algorithm halt, the last pivot is an ordinal pivot and thus the last ordinal pivot gives $B_i=O_i$ for some $i\in\mathbb{Z}_+$. 
\end{proof}

In light of Lemma~\ref{lem:last-ordinal}, it follows that Scarf's algorithm with $(A',b',C')$ will terminate at the end of an iteration (i.e., the algorithm cannot terminate after just the cardinal pivot of an iteration). We remark that the property that the algorithm terminates after an ordinal pivot may not be true in general.

We will henceforth assume that our matrices $A,C$ have such a $0$-th row and column as shown in~\eqref{eq:0-Row-Column}. 

\subsubsection{Separator}

When applied to the hypergraphic preference system, Lemma~\ref{lem:0-Row-Column} implies that without loss of generality, we can add a node $0$ to the instance, with the only hyperedge that contains $0$ being $e_0=\{0\}$. After a stable matching for the new hypergraphic preference system is obtained, removing $\{0\}$ gives a stable matching in the original instance. This modification to the instance, however, changes the node that ``controls'' Scarf's algorithm. In the modified instance, we have $V=\{0,1,\dots,n\}$ and $E=\{e_0,e_1,\dots,e_m\}$ where the hyperedge indices are ordered according to the block-partition of $A,C$, that is, $e_j$ corresponds to the $j$-th column of $A$ (as well as $C$). In particular, for $i\in V$, we use $e_i=\{i\}$ to denote a singleton hyperedge. 





Let $(B,O)$ be a Scarf pair. Suppose the $0$-disliked column w.r.t.~$O$ is $j^\rightarrow$. We recall that by Lemma~\ref{lem:0-Row-Column}, the $0$-th row of $C$ is $(0,m,m-1,\dots,2,1)^T$, which is decreasing except for the $0$-th column. Since $(B,O)$ is a Scarf pair, $0\notin O$. Thus, the $0$-disliked column $j^\rightarrow$, by definition, is the rightmost column in $O$ (i.e.~the column with the largest index). This is the reason for denoting it as $j^{\rightarrow}$.

\begin{definition}\label{def:Separator}
    Let $(B,O)$ be a Scarf pair. Suppose the $0$-disliked column w.r.t.~$O$ is $j^\rightarrow$. If $j^\rightarrow$ belongs to the $i$-th block $S_i$, then we say that $i$ is the \emph{separator} of $(B,O)$.
\end{definition}

Since the blocks form a partition of the non-singleton hyperedges, it follows that the separator is unique. We show that the separator satisfies the following properties:

\begin{lemma}\label{lem:Separator-singleton}
    Let $(B,O)$ be a Scarf pair and suppose $i$ is the separator of $(B,O)$. Let $e_i$ be the column which corresponds to the singleton hyperedge $\{i\}$. Then, $e_i$ belongs to $O$ and $e_i$ is $i$-disliked.
\end{lemma}

\begin{proof}
Suppose $i$ is the separator. Consider the $0$-disliked column $j^\rightarrow$ in $O$. We have $u^O_0=c_{0,j^\rightarrow}$ by definition. By Lemma~\ref{lem:disliked-oneone}, the $i$-disliked column cannot be $j^\rightarrow$, which gives $u^O_i<c_{i,j^\rightarrow}$. Therefore, there exists some $j\in O$ such that $c_{i,j}<c_{i,j^\rightarrow}$.

Let $j\in O$ be such that $c_{i,j}<c_{i,j^\rightarrow}$. Recall the permutation of entries in $C$: $j$ is either in block $i$ or $j=i$ (in which case $c_{i,j}=c_{i,i}=0$). If the former happens, since the entries are decreasing inside the $i$-th block, we must have $j>j^\rightarrow$, which is a contradiction. Therefore, $j=i$ and thus $i\in O$. Since $c_{i,i}=0$, we have $u^O_i=0$ and column $i$ (singleton $e_i$) is $i$-disliked.
\end{proof}

\begin{lemma}\label{lem:Separator-Change}

    Let $(B,O)\to(B',O)\to(B',O')$ be an iteration in Scarf's algorithm. Suppose $j_\ell$ is the leaving column of the ordinal pivot $O\to O'$. Let $j_r$ and $j^*$ be the reference column and the entering column of the ordinal pivot $O\to O'$ respectively. If $j_\ell$ is $i$-disliked w.r.t. $O$ and $i$ is the separator of $(B, O)$, then we have the following:
    \begin{enumerate}
        \item[(i)] $j_r=j^\rightarrow$ is the rightmost column in $O$.
        \item[(ii)] $j_r$ is $0$-disliked w.r.t. $O$ and $j^*$ is $0$-disliked w.r.t. $O'$.
        \item[(iii)] The separator of $(B',O')$ is $i'$ for some $i'>i$.
    \end{enumerate}
    
\end{lemma}

\begin{proof}
Given the hypothesis that $i$ is the separator of $(B,O)$, by Lemma~\ref{lem:Separator-singleton} we know $j_\ell=i$. Also, we have argued in Lemma~\ref{lem:Separator-singleton} that if $j\in O$ such that $c_{i,j}<c_{i,j^\rightarrow}$, then $j=i$. Therefore, $c_{i,j^\rightarrow}$ is the second least entry of row $i$ among $O$. By definition of the reference column, we have $j_r=j^\rightarrow$. This shows (i).

By Definition~\ref{def:Separator}, $j^{\rightarrow}$ is $0$-disliked w.r.t.~$O$, so is $j_r$. By Lemma~\ref{lem:utility} with $i_r=0$, the ordinal pivot $O\to O'$ will introduce a new $0$-disliked column $j^*$ w.r.t. $O'$. This shows (ii).

By Lemma~\ref{lem:op-def}, when choosing the new entering column $j^*$, we need to satisfy that $c_{i',j^*}>u_{i'}^{O-j_\ell}$ for $i'\neq 0$. Thus in particular $c_{i,j^*}>u_i^{O-j_\ell}=u_i^{O-i}=c_{i,j_r}$ (the last equality holds since $c_{i,j_r}=c_{i,j^\rightarrow}$ is the second least entry of row $i$ among $O$). Thus, by the permutation rule of row $i$, either $j^*<j^\rightarrow$, or $j^*>j^\rightarrow$ and $j^*$ is not in block $i$. Suppose $j^*<j^\rightarrow$ happens. If $j^*=0$, then we have $B'=O'$ and the algorithm terminates. If $0<j^*<j^\rightarrow$, then we have $c_{0,j^*}>c_{0,j^\rightarrow}=c_{0,j_r}$, which violates Lemma~\ref{lem:utility}, a contradiction. Thus, if $(B',O')$ is a new Scarf pair, we can claim that $j^*>j^\rightarrow$ and $j^*$ is not in block $i$. Suppose that $j^*$ is in block $i'$, then by the fact $j^*>j^\rightarrow$, we have $i'>i$. By definition~\ref{def:Separator}, since $j^*$ is $0$-disliked w.r.t. $O'$, $i'$ is the new separator. This shows (iii). \end{proof}

\subsection{Feasible Cardinal Bases and Cardinal Pivots for Network Matrices}\label{sec:Pivot-NM}
In this section, we present properties about bases and cardinal pivots of network matrices. 
We start with the following property about network matrices of directed trees. 
\begin{lemma}\label{lem:Base-Inverse}
    Suppose that $\mathcal{T}_1=(U,\mathcal{A}_1)$ and $\mathcal{T}_2=(U,\mathcal{A}_2)$ are two directed trees. Let $A_1$ be the $\mathcal{A}_1\times\mathcal{A}_2$ network matrix corresponding to $(\mathcal{D}_1=(U,\mathcal{A}_2),\mathcal{T}_1)$, and $A_2$ be the $\mathcal{A}_2\times\mathcal{A}_1$ network matrix corresponding to $(\mathcal{D}_2=(U,\mathcal{A}_1),\mathcal{T}_2)$. Then,
    \begin{equation}\label{eq:Base-Inverse}
        A_1\cdot A_2=I_n.
    \end{equation}
\end{lemma}

\begin{proof}
    Without loss of generality, we assume $\mathcal{A}_1\cap\mathcal{A}_2=\emptyset$. Otherwise, we can remove the common arcs from each set, split them into connected components, and then apply the result for each connected component, which corresponds to the same block of matrix $A_1$, $A_2$. Putting the submatrices together gives the desired result.

Let $\mathcal{D}_{1,2}=(U,\mathcal{A}_1\cup\mathcal{A}_2)$ be the directed graph whose arcs are all and only the arcs from $\mathcal{T}_1$ and $\mathcal{T}_2$, denote by $A$ the \emph{incidence matrix} of $\mathcal{D}_{1,2}$,~i.e.,~$A\in\{0,\pm 1\}^{U\times (\mathcal{A}_1\cup\mathcal{A}_2)}$ such that

\begin{equation}\label{eq:Digraph-Matrix}
    A(v,a)=\left\{
    \begin{array}{cc}
        1 & \textrm{if $v$ is the head of $a$;}\\
        -1 & \textrm{if $v$ is the tail of $a$;} \\
        0 & \textrm{otherwise.}
    \end{array}
    \right.
\end{equation}
Denote by $A=(X_1|X_2)$, where the columns of $X_1,X_2$ correspond to the arcs in $\mathcal{A}_1,\mathcal{A}_2$, respectively. One can see that $X_1,X_2$ is the incidence matrix of $\mathcal{T}_1,\mathcal{T}_2$, respectively. We claim the following:

\begin{claim}\label{cl:invert}
    $X_1\cdot A_1=X_2$ and $X_2\cdot A_2=X_1$.
\end{claim}

\noindent \emph{\underline{Proof of Claim~\ref{cl:invert}.}} Indeed, to show the first equation, consider a column $y^a$ in matrix $X_2$ and suppose it corresponds to the arc $a=(u,v)\in \mathcal{A}_2$. Then, $y^a$ has only two nonzeros,~i.e.~$y^a_u=1$ and $y^a_v=-1$. Consider the linear system $X_1\cdot z=y^a$. Here, the solution $z$ represents a feasible unit flow sent from $u$ (as a supply) to $v$ (as a demand) in the directed graph $\mathcal{T}_1=(U,\mathcal{A}_1)$ since $X_1$ is the incidence matrix of $\mathcal{T}_1$. Since $\mathcal{T}_1$ is a tree on $U$, the only way to send a flow from $u$ to $v$ is to find the unique $\mathcal{T}_1$-path from $u$ to $v$, and assign $1$ unit flow on the forward arcs and $-1$ unit flow on the backward arcs. Thus, there is a unique solution $z$ to $X_1\cdot z=y^a$ such that $z^{a'}\neq 0$ if and only if $a'$ appears in the $\mathcal{T}_1$-path from $u$ to $v$, and is $1$ (resp.~$-1$) if $a'$ is forward (resp.~backward). Compare with the definition of~\eqref{eq:Network-Matrix}, we have that the feasible solution $z$ is exactly the column vector of $A_1$ that corresponds to $a$. This argument applies to every $a\in\mathcal{A}_2$, hence $X_1\cdot A_1=X_2$. Similarly it holds that $X_2\cdot A_2=X_1$. \hfill $  \blacksquare$ 

\smallskip

We notice that both $X_1$ and $X_2$ have rank $|U|-1$, as they are the incidence matrix of a directed tree. Let $X_1'$ and $X_2'$ be obtained from $X_1$ and $X_2$ by deleting the same row from both, respectively. Then $X_1'$ and $X_2'$ are both invertible matrices with
$$X_1'\cdot A_1=X'_2\textrm{ and }X_2'\cdot A_2=X_1'.$$
Thus, we have $A_1\cdot A_2=(X_1')^{-1} X_2'(X_2')^{-1}X_1'=I_n$.
\end{proof}

From Lemma~\ref{lem:Base-Inverse}, we deduce that every cardinal basis for a network matrix corresponds to a directed tree on $U$.

\begin{theorem}[\cite{schrijver1998theory}]\label{lem:Base-SpanningTree}
Let $\mathcal{D}=(U,\mathcal{A})$ be a directed graph with $|U|=n$, $|\mathcal{A}|=m$ and $\mathcal{T}=(U,\mathcal{A}_0)$ be a directed tree. Let $A$ be the $\mathcal{A}\times \mathcal{A}_0$ network matrix corresponding to $(\mathcal{D},\mathcal{T})$. Consider a subset $B\subset [m]$ of $n$ columns. Denote the corresponding arc set by $\mathcal{A}_B$. Then the following statements are equivalent:
    \begin{enumerate}
        \item[(i)] $B$ is a cardinal basis. 
        \item[(ii)] $\mathcal{T}_B=(U,\mathcal{A}_B)$ is a directed tree.
    \end{enumerate}
\end{theorem}

\begin{proof}
    If $B$ is a basis, then $|B|=n$. If $(U,\mathcal{A}_B)$ is not a tree, then there exists a cycle in $\mathcal{A}_B$. Suppose that the arcs $a_{j_1}\to\dots\to a_{j_k}$ corresponding to the columns $j_1,\dots,j_k\in B$ create a cycle. Then, let $A_{j_1},\dots,A_{j_k}$ be the corresponding column vectors of $A$. Begin with the tail of $a_{j_1}$, walk along the cycle through the arcs $a_{j_1}\to\dots\to a_{j_k}$, define $\lambda_i=1$ if $a_{j_i}$ goes forward and $\lambda_i=-1$ if $a_{j_i}$ goes backward. By definition of network matrix, $\lambda_1 A_{j_1}+\dots+\lambda_k A_{j_k}=0$. Therefore, $A_{j_1},\dots,A_{j_k}$ are linearly dependent, a contradiction. Thus $(U,\mathcal{A}_B)$ is a directed tree.

    Conversely, if $\mathcal{T}_B=(U,\mathcal{A}_B)$ is a directed tree, the submatrix $A_B$ is a network matrix corresponding to $(\mathcal{T},\mathcal{T}_B)$. By Lemma~\ref{lem:Base-Inverse}, $A_B$ is invertible, thus $B$ is a basis.
\end{proof}

\medskip


    

Next, we will give a combinatorial interpretation of the cardinal pivot for a network matrix. We need the following lemma. 

\begin{lemma}\label{lem:Forward-Backward-Formula}
Let $\mathcal{D}=(U,\mathcal{A})$ be a directed graph with $|U|=n$, $|\mathcal{A}|=m$ and $\mathcal{T}=(U,\mathcal{A}_0)$ be a directed tree. Let $A$ be the $\mathcal{A}\times \mathcal{A}_0$ network matrix corresponding to $(\mathcal{D},\mathcal{T})$, $b=1^V$, and let $P:=\{x\in \R^{\mathcal{A}_0}_{\ge 0}: Ax=b \}$. 
    Let $B$ be a cardinal basis of $A$ and suppose $j_t\notin B$. 
    Let $a_{j_t}=(v,v')$ be the arc corresponding to the column $j_t$. Let $P_B(v,v')$ be the $\mathcal{T}_B$-path from $v$ to $v'$ where $\mathcal{T}_B=(U,\mathcal{A}_B)$. Define $y:=A_B^{-1}\cdot A_{j_t}$. Then, we have

    \begin{equation}\label{eq:Forward-Backward-Formula}
    y_a=\left\{
    \begin{array}{cc}
        1 & \textrm{if $a\in\mathcal{A}_B$ is a forward arc on path $P_B(v,v')$,}\\
        -1 & \textrm{if $a\in\mathcal{A}_B$ is a backward arc on path $P_B(v,v')$, and} \\
        0 & \textrm{if $a$ does not belong to the path $P_B(v,v')$.}
    \end{array}
    \right.
\end{equation}
\end{lemma}

\begin{proof}
    For simplicity, we abuse the notation by denoting $A_B^{-1}=(\beta_{ij})_{i,j\in[n]}$ and $A_{j_t}=(\gamma_j)_{j\in[n]}$. Also, to distinguish the arcs in $\mathcal{A}_0$ and $\mathcal{A}_B$ we let $f_i,a_j$ be the representative arcs in $\mathcal{A}_0,\mathcal{A}_B$, respectively. 

    Since $\mathcal{T}=(U,\mathcal{A}_0)$ is the principal tree and $B$ is a basis, by Lemma~\ref{lem:Base-Inverse}, the network matrix $A_2$ defined by $\mathcal{D}_2=(U,\mathcal{A}_0,\mathcal{T}_B)$ is exactly $A_B^{-1}$. Therefore we have $y=A_2\cdot A_{j_t}$.

     First, suppose that the arc $a_{j_t}=(f_{i_1},f_{i_2},\dots,f_{i_q})$ is composed by the unique $\mathcal{T}$-path $P_0(v,v')=(f_{i_1},f_{i_2},\dots,f_{i_q})$ from $v$ to $v'$. We abuse the notation by writing $f_k\in P_0(v,v')$ if $f_k\in\{f_{i_1},f_{i_2},\dots,f_{i_q}\}$. Notice that $\gamma_k\neq 0$ if and only if $f_k\in P_0(v,v')$ by~\eqref{eq:Network-Matrix}. Thus, we have 
    \begin{align}\label{eq:row-summation}
        y_i=\sum_{f_k\in P_0(v,v')}\beta_{ik}\gamma_k&=\sum_{f_k\in P_0(v,v'),\gamma_k=1}\beta_{ik}-\sum_{f_k\in P_0(v,v'),\gamma_k=-1}\beta_{ik}\nonumber\\
        &=\sum_{f_k\in P_0(v,v')\textrm{, $f_k$ forward}}\beta_{ik}-\sum_{f_k\in P_0(v,v')\textrm{, $f_k$ backward}}\beta_{ik}.
    \end{align}
    For every $s\in[q]$, we define $P_B^s$ as the unique $\mathcal{T}_B$-path generated by $f_{i_s}$,~i.e.~one can define $P_B^s$ as the sequence $P_B^s=(a_{k_1},a_{k_2},\dots,a_{k_\ell})$, where $P_B^s$ starts from the tail of $f_{i_s}$ and ends at the head of $f_{i_s}$. By definition, since $A_2=A_B^{-1}$ is a network matrix corresponding to $(\mathcal{T},\mathcal{T}_B)$, $\beta_{ik}$ indicates the direction of the arc $a_i\in\mathcal{A}_B$ on $P_B^s$. Therefore, $\beta_{ik}=1$ (resp.~$-1$) if there exists $s\in[q]$ such that $k=i_s$, and $a_i$ is a forward (resp.~backward) arc on $P_B^s$. From~\eqref{eq:row-summation}, we obtain

    \begin{align}\label{eq:row-summation-decompose}
        y_i=&\sum_{s\in[q],\textrm{ $f_{i_s}$ forward}}\left(\mathbf{1}(a_i\in P_B^s,\textrm{$a_i$ forward})-\mathbf{1}(a_i\in P_B^s,\textrm{$a_i$ backward})\right)\nonumber\\
        &-\sum_{s\in[q],\textrm{ $f_{i_s}$ backward}}\left(\mathbf{1}(a_i\in P_B^s,\textrm{$a_i$ forward})-\mathbf{1}(a_i\in P_B^s,\textrm{$a_i$ backward})\right). 
    \end{align}

    Let $\bar{P}_B^s=P_B^s$ if $f_{i_s}$ is forward on $a_{j_t}$, and $\bar{P}_B^s$ is the reverse (reversing the order of the arcs in the sequence, for example, if $P_B^s=(a_{k_1},a_{k_2},\dots,a_{k_\ell})$, then $\bar{P}_B^s=(a_{k_\ell},\dots,a_{k_2},a_{k_1})$) of $P_B^s$ if $f_{i_s}$ is backward on $a_{j_t}$. Then~\eqref{eq:row-summation-decompose} becomes

    \begin{equation}\label{eq:row-summation-counting}
        y_i=\sum_{s\in[q]}(\mathbf{1}(a_i\in \bar{P}_B^s,\textrm{$a_i$ forward})-\mathbf{1}(a_i\in \bar{P}_B^s,\textrm{$a_i$ backward})).
    \end{equation}

    Now consider a walk $W=\bar{P}_B^1\oplus\cdots\oplus \bar{P}_B^q$ (see Definition~\ref{def:walk}). 
    Then according to~\eqref{eq:row-summation-counting}, $y_i$ counts exactly the difference between the number of times $a_i$ appears forward on $W$ and that of times $a_i$ appears backward on $W$. Notice that our goal is to relate $y$ with the path $P_0(v,v')$, where such path $P_0(v,v')$ is a reduced version of $W$ without redundant cycles. Notice that $P_0(v,v')$ can be obtained by deleting from $W$ the arcs that appear multiple times, and when deleting them, we cancel the forward and backward arcs pair by pair, which results in that the difference between the number of occurrences maintain the same. At the end, each arc $a_i$ either goes forward or backward, thus we have exactly the meaning shown in~\eqref{eq:Forward-Backward-Formula}.
\end{proof}

As a consequence of Lemma~\ref{lem:Forward-Backward-Formula}, we show that the leaving column of a cardinal pivot should necessarily be a forward arc in the unique path from the tail of the entering arc to the head of the entering arc. 
\begin{corollary}\label{lem:Leaving-Forward}
Let $\mathcal{D}=(U,\mathcal{A})$ be a directed graph with $|U|=n$, $|\mathcal{A}|=m$ and $\mathcal{T}=(U,\mathcal{A}_0)$ be a directed tree. Let $A$ be the $\mathcal{A}\times \mathcal{A}_0$ network matrix corresponding to $(\mathcal{D},\mathcal{T})$, $b=1^V$, and $P:=\{x\in \R^{\mathcal{A}_0}_{\ge 0}: Ax=b \}$. 
Let $B$ be a feasible cardinal basis and let $\mathcal{T}_B=(U,\mathcal{A}_B)$. 
     Consider the cardinal pivot from $B$ with entering column $j_t$. Let $a_{j_t}=(v,v')$ be the arc corresponding to column $j_t$. Let $P_B(v,v')$ be the unique $\mathcal{T}_B$-path from $v$ to $v'$ (unique by Theorem \ref{lem:Base-SpanningTree}). Then, $j_\ell$ is a leaving candidate of the cardinal pivot only if the arc $a_{j_\ell}$ is a forward arc on $P_B(v,v')$.
\end{corollary}

\begin{proof}
By Lemma~\ref{lem:Forward-Backward-Formula}, $y_a>0$ if and only if $a$ is a forward arc on $P_B(v,v')$. Thus, by Definition~\ref{def:cpdef}, if $j_\ell$ is a leaving candidate, it has to satisfy $y_{a_{j_\ell}}>0$, which implies that $a_{j_\ell}$ is an arc that is forward on $P_B(v,v')$.
\end{proof}

In order to give a combinatorial interpretation of non-degenerate cardinal pivot for a network matrix, we will use the notion of augmenting and descending paths. 

\begin{definition}[Augmenting and Descending Paths]\label{def:x-alternating}
    Let $\mathcal{D}=(U,\mathcal{A})$, $x\in\mathbb{R}^\mathcal{A}$, and $P=(a_{j_1},\dots,a_{j_p})$ be a $\mathcal{D}$-path. 
    \begin{enumerate}
        \item The path $P$ is \emph{$x$-augmenting} if every forward arc $a_{fwd}$ on $P$ has $x_{a_{fwd}}=1$ and every backward arc $a_{bwd}$ on $P$ has $x_{a_{bwd}}=0$.
        \item The path $P$ is \emph{$x$-descending} if every forward arc $a_{fwd}$ on $P$ has $x_{a_{fwd}}=0$ and every backward arc $a_{bwd}$ on $P$ has $x_{a_{bwd}}=1$.
    \end{enumerate}
    
\end{definition}
We show that a non-degenerate cardinal pivot is characterized by an $x$-augmenting path.


\begin{theorem}\label{lem:Nonde-Pivot}
Let $\mathcal{D}=(U,\mathcal{A})$ be a directed graph with $|U|=n$, $|\mathcal{A}|=m$ and $\mathcal{T}=(U,\mathcal{A}_0)$ be a directed tree. Let $A$ be the $\mathcal{A}\times \mathcal{A}_0$ network matrix corresponding to $(\mathcal{D},\mathcal{T})$, $b=1^V$, and $P:=\{x\in \R^{\mathcal{A}_0}_{\ge 0}: Ax=b \}$. 
Let $B$ be a feasible cardinal basis and let $\mathcal{T}_B=(U,\mathcal{A}_B)$. 
Consider a cardinal pivot from $B$ to $B'$ with $j_t$ as the entering column where $a_{j_t}=(v,v')\in\mathcal{A}$. Let $x,x'$ be the extreme points corresponding to $B,B'$, respectively. Let $P_B(v,v')$ be the unique $\mathcal{T}_B$-path from $v$ to $v'$ (unique by Theorem \ref{lem:Base-SpanningTree}). The following statements are equivalent:

    \begin{enumerate}
        \item[(i)] The cardinal pivot is non-degenerate.
        \item[(ii)] $x_{j_t}=0$ and $x'_{j_t}=1$.
        \item[(iii)] $P_B(v,v')$ is $x$-augmenting.
        
        \item[(iv)] $P_B(v,v')$ is $x'$-descending.
    \end{enumerate}

\end{theorem}
\begin{proof}
If the cardinal pivot is non-degenerate, then we have that $x\neq x'$. We claim that $x'_{j_t}\neq 0$. Indeed, by the equation $A_Bx_B+A_{j_t}x_{j_t}=A_Bx'_B+A_{j_t}x'_{j_t}=b$, if $x'_{j_t}=0$, then since $A_B$ has full rank, we must have $x_B=x'_B$, which results in $x=x'$, a contradiction. Since $x'$ is feasible, we have $x'_{j_t}>0$. By the fact that $A$ is totally unimodular, we have $x'_{j_t}=1$. Thus, (i) implies (ii). (ii) obviously implies (i) since by definition $x\neq x'$, thus the cardinal pivot is non-degenerate. Therefore, we obtain (i)$\Leftrightarrow$(ii).

    
By definition of pivoting (c.f.~Section~\ref{sec:pre-cp}), we have $A_Bx_B=1$ and $A_Bx_B+A_{j_t}x_{j_t}=1$. Thus we have $A_Bx'_B+A_{j_t}x'_{j_t}=1$, which gives $x'_B=x_B-A_B^{-1}A_{j_t}x'_{j_t}=x_B-x'_{j_t}y$ where $y$ is given by~\eqref{eq:Forward-Backward-Formula}. 
    
Now, assume (ii) holds. Observe that $x'_{j_t}=1$ and $x'_B\ge 0$ means for all $a\in\mathcal{A}_B$, $y_a=1$ implies $x_a=1$. By Lemma~\ref{lem:Forward-Backward-Formula}, $y_a=1$ if and only if $a$ is a forward arc on $P_B(v,v')$, thus every forward arc $a_{fwd}$ has $x_{a_{fwd}}=1$ and $x'_{a_{fwd}}=0$. On the other hand, consider any backward arc $a_{bwd}$ on $P(v,v')$, if $x_{a_{bwd}}=1$, then $x'_{a_{bwd}}=1-1\times(-1)=2$, a contradiction with feasibility of $x'$. Therefore, $x_{a_{bwd}}=0$ and $x'_{a_{bwd}}=1$. Thus, we obtain (ii)$\Rightarrow$(iii) and (ii)$\Rightarrow$(iv).

On the other hand, if (iii) holds, then by Lemma~\ref{lem:Forward-Backward-Formula} and $A_Bx_B+A_{j_t}x_{j_t}=1$, we obtain (ii). If (iv) holds, then again by Lemma~\ref{lem:Forward-Backward-Formula} and $A_Bx'_B+A_{j_t}x'_{j_t}=1$, we can deduce (ii). Therefore, we obtain (iii)$\Rightarrow$(ii) and (iv)$\Rightarrow$(ii).
\end{proof}


\subsection{Polynomiality of Scarf's algorithm}
In this section, we describe a pivoting rule and establish polynomiality of Scarf's algorithm using this pivoting rule for every arborescence hypergraphic preference system. 
\subsubsection{Depth-first Order}\label{sec:DFO}


In order to guarantee some nice properties on the cardinal bases (or equivalently, directed trees) visited by the algorithm, we permute the vertices and arcs of the principal tree according to a depth-first order. This order is a common topological order on trees. Details can be found in~\cite{cormen2022introduction}. 
\begin{definition}\label{def:depth-first}
    Let $\mathcal{T}=(U,\mathcal{A}_0)$ be an arborescence with root $r$. We define a partial order $\ge_\mathcal{T}$ over $U$ such that $v\ge_\mathcal{T} v'$ if the $\mathcal{T}$-directed path from $r$ to $v'$ passes through $v$. We say that $\mathcal{T}$ is \emph{depth-first} if
   \begin{enumerate}
    \item $U=\{v_1,\dots,v_{n+1}\}$, and for $i,j\in [n+1]$, $v_i\ge_\mathcal{T} v_j$ implies $i\ge j$. In particular, $v_{n+1}$ is the unique root of $\mathcal{T}$.
    \item $\mathcal{A}_0=\{f_1,\dots,f_n\}$, and for every $i\in [n]$, the head of $f_i$ is $v_i$.
   \end{enumerate}
\end{definition}
For example, the arborescence in Figure~\ref{fig:clean_order} is depth-first. For every given arborescence $\mathcal{T}=(U,\mathcal{A}_0)$ with $|U|=n+1$, a permutation of $U$ and $\mathcal{A}_0$ such that $\mathcal{T}$ is depth-first ordered can be found in time $O(n)$. 
\begin{figure}[h!]
    \centering
    \includegraphics[scale=1.2]{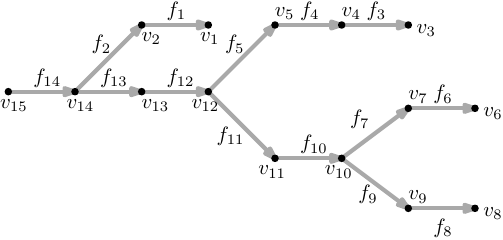}
    \caption{A depth-first arborescence $\mathcal{T}=(U,\mathcal{A}_0)$. $U=\{v_1,\dots,v_{15}\}$ and $r=v_{15}$ is the root. $\mathcal{A}_0=\{f_1,\dots,f_{14}\}$. Notice that for every $i\in[14]$, the head of $f_i$ is $v_i$, yet the tail of $f_i$ may not be $v_{i+1}$ (for example, $f_2=(v_{14},v_2)$).}
    \label{fig:clean_order}
\end{figure}

\subsubsection{The Pivoting Rule}\label{sec:FFL}

We have seen the interpretation of cardinal pivot for network matrices in Section~\ref{sec:Pivot-NM}. When applied to the incidence matrix of an arborescence hypergraph corresponding to $(\mathcal{D}=(U,\mathcal{A}),\mathcal{T}=(U,\mathcal{A}_0))$, a variable $x_a$ represents a specific arc $a\in \mathcal{A}$. From this perspective, by saying that the entering arc (resp.~leaving arc) is $a_{j_t}$ (resp.~$a_{j_\ell}$), we mean the entering variable/column (resp.~leaving variable/column) is $x_{j_t}$/$j_t$ (resp.~$x_{j_\ell}$/$j_\ell$). 


Let $B$ be a feasible cardinal basis and consider a cardinal pivot from $B$. Suppose that we have $a_{j_t}$ as an entering arc and $a_{j_t}=(v,v')$. Then, there is a $\mathcal{T}_B$-path from $v$ to $v'$, denoted by $P_B(v,v')$. By Corollary~\ref{lem:Leaving-Forward}, the leaving arc $a_{j_\ell}$ has to be a forward arc on $P_B(v,v')$. We will use the following \emph{first forward arc leaving (FFL) rule} as the pivoting rule. 

\begin{definition}[First forward arc leaving rule]\label{def:First-Forward-Leaving}
Let $B$ be a feasible cardinal basis. Let $a_{j_t}=(v,v')$ be the entering arc and $P_B(v,v')=(\bar{a}_1,\cdots,\bar{a}_p)$ be the unique $\mathcal{T}_B$-path. If there exists a forward arc $\bar{a}_k$ in $P_B(v,v')$ with $x_{\bar{a}_k}=0$, choose the one with smallest subscript $k$ as the leaving arc. If all forward arcs $\bar{a}_{fwd}$ in $P_B(v,v')$ have $x_{\bar{a}_{fwd}}=1$, then let $k$ be the smallest index such that $\bar{a}_k$ is forward on $P_B(v,v')$, and let $\bar{a}_k$ be the leaving arc. 
\end{definition}

We first justify that FFL rule is valid, meaning that for every entering variable $j_t\notin B$, (i) the leaving arc defined above exists and is unique, and (ii) suppose that $j_\ell$ is the leaving variable defined by the FFL rule, then $j_\ell$ is indeed a leaving candidate (c.f.~Definition~\ref{def:degenerate}).

\begin{lemma}\label{lem:FFL-well-defined}
    The FFL rule is a valid pivoting rule.
\end{lemma}
\begin{proof}
We begin by observing that for every $a=(u,u')\in\mathcal{A}$, we have that $u\ge_{\mathcal{T}}u'$. Let $a=(u, u')\in \mathcal{A}$. By definition, there exists a $\mathcal{T}$-directed path $P_0(u,u')$ from $u$ to $u'$. Let $r$ be the root of the arborescence $\mathcal{T}$, then the $\mathcal{T}$-path $P_0(r,u)$ from $r$ to $u$ has only forward arcs. Therefore, $P_0(r,u)\oplus P_0(u,u')$ also has only forward arcs, which implies that $P_0(r,u)\oplus P_0(u,u')$ is a $\mathcal{T}$-directed path from $r$ to $u'$. Thus, $u\ge_{\mathcal{T}} u'$. 

 Let $a_{j_t}=(v,v')\in\mathcal{A}$ be the entering arc. 
In order to show validity of the pivoting rule, we first show that there exists at least one forward arc in $P_B(v,v')$. By the observation in the previous paragraph, we have that $v\ge_\mathcal{T}v'$. For the sake of contradiction, suppose that $P_B(v,v')=(\bar{a}_1,\dots,\bar{a}_p)$ only contains backward arcs. Let $\bar{a}_i=(\bar{v}_i,\bar{u}_i)$ for every $i\in[p]$. 
Since $\bar{a}_i\in\mathcal{A}$, we have that $\bar{v}_i\ge_{\mathcal{T}}\bar{u}_i$ for every $i\in [p]$. Thus, we have $v'=\bar{v}_p\ge_{\mathcal{T}}\bar{u}_p=\bar{v}_{p-1}\ge_{\mathcal{T}}\dots\ge_{\mathcal{T}}\bar{u}_1=v$. This contradicts with $v\ge_{\mathcal{T}} v'$. Consequently, there exists at least one forward arc in $P_B(v,v')$. By the definition of the FFL rule, the choice of $\bar{a}_k$ is unique. 

Now, in order to show validity of the pivoting rule, it suffices to show that if $a_{j_\ell}=\bar{a}_k$, then $j_\ell$ is a leaving candidate. In order to show this, we need to show that $y_{\ell}>0$ and $\ell\in\arg\min\{x_j/y_j:j\in B,y_j>0\}$. By Lemma~\ref{lem:Forward-Backward-Formula}, we have $y_{\ell}>0$ (and indeed $y_a>0$ if and only if $a$ is a forward arc on $P_B(v,v')$). Thus, $j_\ell$ is always the minimizer as defined above, since either $x_{j_\ell}=0$, or $x_{j_\ell}=1$ and every $j\in B$ such that $y_j>0$ satisfies $x_j/y_j=1$ according to the FFL rule.
\end{proof}
We rewrite the cardinal pivot operation with FFL rule in Algorithm~\ref{alg:cp}.
\begin{algorithm}
\caption{Cardinal Pivot with FFL Rule}\label{alg:cp}
\begin{algorithmic}
\State{Let $B$ be the current feasible cardinal basis, associated with basic feasible solution $x$. Suppose that $j_t\notin B$ is the entering column.}
\State{Find $a_{j_t}=(v,v')$ and the $\mathcal{T}_B$-path $P_B(v,v')=(\bar{a}_1,\dots,\bar{a}_p)$}
\State{$I\gets \{i\in[p],\bar{a}_i\textrm{ is forward and }x_{\bar{a}_i}=0\}$}
\State{$J\gets \{i\in[p],\bar{a}_i\textrm{ is forward}\}$}
\If {$I\neq \emptyset$}
\State{$k\gets\min \{i:i\in I\}$}
\Else 
\State{$k\gets\min \{i: i\in J\}$}
\EndIf
\State{Let the leaving column $j_\ell$ be such that $a_{j_\ell}=\bar{a}_k$.}
\State{$B\gets B\cup\{j_t\}\setminus\{j_\ell\}$}
\end{algorithmic}
\end{algorithm}

\subsubsection{Proof of Theorem~\ref{thm:main-arb-poly}}\label{sec:run-time}



From now on, we assume the following: the arborescence hypergraph $H=(V,E)$ with principal tree $\mathcal{T}=(U,\mathcal{A}_0)$ where $U=\{v_1,\dots,v_{n+1}\}$ and $\mathcal{A}_0=\{f_1,\dots,f_n\}$ are depth-first (c.f.~Definition~\ref{def:depth-first}). The matrices $A,C$ as the input of Scarf's algorithm are block-partitioned (c.f.~Definition \ref{def:block}) with additional $0$-th row and column as in Lemma~\ref{lem:0-Row-Column}. Every cardinal pivot follows the FFL rule (c.f.~Definition~\ref{def:First-Forward-Leaving}). We will analyze Scarf's algorithm and show Theorem~\ref{thm:main-arb-poly} under these assumptions.

Let $(B_0,O_0)$ be the initial Scarf pair of the algorithm where $B_0=\{0,1,\dots,n\}$ and $O_0=\{1,2,\dots,n,n+1\}$. The algorithm executes alternatively a cardinal pivot with the FFL rule and an ordinal pivot. Then, we obtain a unique sequence~\eqref{eq:iterations} and each $(B_k,O_k)$ in that sequence is a Scarf pair. We will maintain $|B_k\cap O_k|\ge n$, and if $B_k\neq O_k$, then $0=B_k-O_k$. We now define a notion of a well-structured basis and will subsequently show that all cardinal bases visited by the algorithm 
are well-structured. 
This structural property of the basis will be helpful in bounding the number of iterations of the algorithm. 


\begin{definition}\label{def:nice-basis}
Let $B$ be a feasible cardinal basis and $\mathcal{T}_B=(U,\mathcal{A}_B)$ be the directed tree corresponding to $B$. Let $x$ be the extreme point associated with $B$. Let $i\in[n]$. We say that $B$ is an \emph{$i$-nice basis} if the following properties hold:
    \begin{enumerate}
        \item[(i)] Let $v\in U\setminus\{v_{n+1}\}$, and $P_B(v_{n+1},v)$ be the $\mathcal{T}_B$-path from $v_{n+1}$ to $v$. Then, $P_B(v_{n+1},v)$ is $x$-augmenting. 
        \item[(ii)] $\{f_i,f_{i+1},\dots,f_n\}\subset\mathcal{A}_B$. 
        \item[(iii)] Let $f\in\{f_i,f_{i+1},\dots,f_n\}$. Let $\mathcal{T}^{-f}=(U,\mathcal{A}_0-\{f\})$ be the subgraph of $\mathcal{T}$ with exactly two connected components (in the undirected sense). Denote by $U=R\cup W$ the partition of vertices of the two components such that $r\in R$ and $r\notin W$. If an arc $a\in\mathcal{A}_B$ has its end vertices in $R$ and $W$, then $a=f$. In other words, the removal of each arc $f\in \{f_i, f_{i+1}, \ldots, f_n\}$ from $\mathcal{T}_B$ partitions $U$ into the same sets as the removal of the same arc from $\mathcal{T}$.
        
    \end{enumerate}
\end{definition}
See Figure \ref{fig:Pivoting-Ex-B} for an example of an $i$-nice basis. 
The following lemma is the main result of the section. It will allow us to bound the number of iterations in the algorithm. 
\begin{lemma}\label{lem:Out-Deg-1}
Let $(B,O)$ be a Scarf pair and we consider the iteration (a cardinal pivot and an ordinal pivot) $(B,O)\to(B',O)\to(B',O')$. Suppose $j_t=O-B$ is the entering column of the cardinal pivot, $j_\ell=B'-B$ is the leaving column of the cardinal pivot, and $j^*=O'-O$ is the entering column of the ordinal pivot. Let $i\in[n]$ be the separator of $(B,O)$. Suppose $B$ is an $i$-nice basis and $j^*\neq 0$. Then, we have the following: 
\begin{enumerate}
    \item\label{it:key-1} Let $a_{j_t}=(v_j,v_k)$. Then, $f_i=(v_j,v_i)$ (recall that $f_i$ is the unique arc entering $v_i$ in $\mathcal{T}$). In other words, $a_{j_t}$ and $f_i$ share the same tail. In addition, $f_i$ is the first arc in $P_0(v_j,v_k)$, where $P_0(v_j,v_k)$ be the $\mathcal{T}$-path from $v_j$ to $v_k$.
    \item\label{it:key-2} $a_{j_\ell}=f_i$.
    \item\label{it:key-3} The separator of $(B',O')$ is $i'$ for some $i'>i$.
    \item\label{it:key-4} $B'$ is an $i'$-nice basis.
\end{enumerate}
\end{lemma}

\begin{figure}[h!]
    \centering
    \includegraphics[scale=1.2]{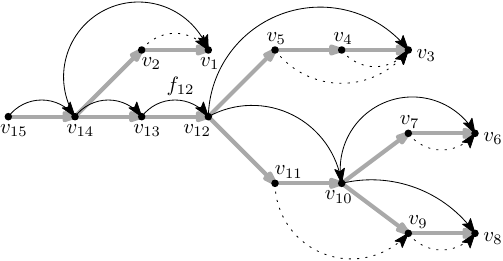}
    \caption{An example of of a basis $B$ that is a $12$-nice basis. The grey arcs are the arcs in  $\mathcal{T}=(U,\mathcal{A}_0)$ while the black arcs are the arcs in $\mathcal{T}_B=(U,\mathcal{A}_B)$. We note that $\mathcal{T}$ is the depth-first arborescence from Figure~\ref{fig:clean_order} and $\mathcal{T}_B$ is associated with a feasible cardinal basis $B$. 
    The solid circular arcs are associated to variables among $B$ with $x$-value $1$ and dotted circular arcs to variables among $B$ with $x$-value $0$. 
    }
    \label{fig:Pivoting-Ex-B}
\end{figure}

\begin{figure}[h!]
    \centering
    \includegraphics[scale=1.2]{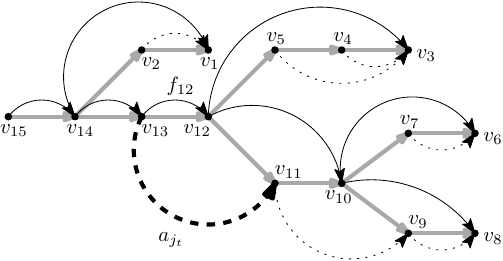}
    \caption{An example of cardinal pivot from the basis $B$ given in Figure~\ref{fig:Pivoting-Ex-B} with entering arc $a_{j_t}=(v_{13},v_{11})$ present by a dashed circular arc. This graph is $(U,\mathcal{A}_B\cup\{a_{j_t}\})$. The $\mathcal{T}_B$-path from $v_{13}$ to $v_{11}$ is $P_B(v_{13},v_{11})=((v_{13},v_{12}),(v_{12},v_{10}),(v_{10},v_{8}),(v_{9},v_{8}),(v_{11},v_9))$, where the first three arcs are forward and the last two arcs are backward.}
    \label{fig:Pivoting-Ex-jt}
\end{figure}

\begin{figure}[h!]
    \centering
    \includegraphics[scale=1.2]{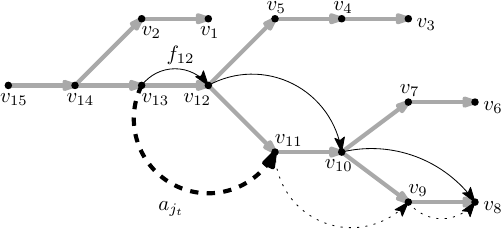}
    \caption{The cycle $F$ created by $a_{j_t}=(v_{13},v_{11})$ and $P_B(v_{13},v_{11})$ from Figure \ref{fig:Pivoting-Ex-jt}. We note that $f_{12}$ is the first arc in $P_B(v_{13},v_{11})$. By the FFL rule, the arc $f_{12}$ will leave $\mathcal{A}_B$.}
    \label{fig:Pivoting-Ex-cycle}
\end{figure}

\begin{figure}[h!]
    \centering
    \includegraphics[scale=1.2]{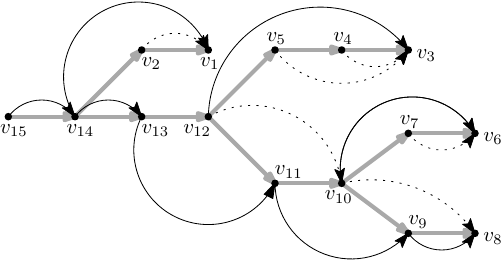}
    \caption{The feasible cardinal basis $B'$ after the cardinal pivot $B\to B'$. We note that the leaving arc $a_{j_\ell}=f_{12}$ is removed from the basis.}
    \label{fig:Pivoting-Ex-B'}
\end{figure}
\begin{proof}
We use the notation in Definition~\ref{def:nice-basis}.
\begin{enumerate}
    \item By Definition~\ref{def:Separator}, let $j^\rightarrow$ be the $0$-th disliked column w.r.t.~$O$. Then, $j^\rightarrow$ is in $S_i$, namely the $i$-th block of $C$. By the block structure, the hyperedge $e_{j^\rightarrow}\in E$ has $i\in e\subseteq\{1,\dots,i\}$. Suppose $a_{j^\rightarrow}=(w,w')$. Then, the $\mathcal{T}$-directed path $P_0(w,w')$ from $w$ to $w'$ contains $f_i$ and only contains arcs among $\{f_1,\dots,f_i\}$. Since $\mathcal{T}$ is depth-first ordered, $P_0(w,w')$ visits a sequence of arcs with subscripts in a descending order, thus the first arc is $f_i$. Therefore, $w=v_j$ where $(v_j,v_i)=f_i$. Also, $w'\neq v_i$ since $e_{j^\rightarrow}\neq e_i$ (notice that $j^\rightarrow$ is in the $i$-th block of $C$, but $i$ is among the first $n+1$ columns of $C$). We claim that $j^\rightarrow=O-B$ and thus $j_t=j^\rightarrow$. In fact, denote by $U=R_i\cup W_i$ the partition of vertices of the two components in $\mathcal{T}^{-f_i}=(U,\mathcal{A}_0-\{f_i\})$ as defined in Definition~\ref{def:nice-basis}(iii). Then, we have $w=v_j\in R_i$ and $w'\in W_i$ since $P_0(w,w')$ contains $f_i$. Therefore, $a_{j^\rightarrow}$ has its end vertices in $R_i$ and $W_i$. If $j^\rightarrow\in B$, then $a_{j^\rightarrow}=f_i$ by the $i$-nice property (iii) of $B$, which contradicts with $e_{j^\rightarrow}\neq e_i$. Hence, $j^\rightarrow=j_t$ since $j^\rightarrow\in O$ and $j^\rightarrow\notin B$. Notice that $a_{j^\rightarrow}$ satisfies the the properties mentioned as conclusions of the lemma, and hence, $a_{j_t}$ also satisfies these properties. 
    \item Let $P_B(v_j,v_k)$ be the $\mathcal{T}_B$-path from $v_j$ to $v_k$ (notice that $a_{j_t}=(v_j,v_k)$). We first claim that
    \begin{equation}\label{eq:vj-vk}
        P_B(v_j,v_k)=f_i\oplus P_B(v_i,v_k),\textrm{ and}
    \end{equation}
    \begin{equation}\label{eq:r-vk}
        P_B(v_{n+1},v_k)=P_B(v_{n+1},v_j)\oplus P_B(v_j,v_k).
    \end{equation}
     In fact, as claimed before, since $v_j\in R_i$ and $v_k\in W_i$, by $i$-nice property (iii) of $B$, $P_B(v_j,v_k)$ contains $f_i$. Also, if $f_i$ is not the first arc in $P_B(v_j,v_k)$, we visit $v_j$ twice through $P_B(v_j,v_k)$, a contradiction. Therefore, equation \eqref{eq:vj-vk} holds. Again, since $v_{n+1}\in R_i$ and $v_k\in W_i$, by $i$-nice property (iii) of $B$ we deduce that the path $P_B(v_{n+1},v_k)$ passes through $f_i=(v_j,v_i)$. Thus, $P_B(v_{n+1},v_k)=P_B(v_{n+1},v_j)\oplus f_i\oplus P_B(v_i,v_k)=P_B(v_{n+1},v_j)\oplus P_B(v_j,v_k)$, which gives~\eqref{eq:r-vk}. Now, by $i$-nice property (i) of $B$, since $P_B(v_{n+1},v_k)$ is $x$-augmenting, the subpath $P_B(v_j,v_k)$ is also $x$-augmenting. By~\eqref{eq:vj-vk}, we know that $f_i$ is the first arc on the $x$-augmenting path $P_B(v_j,v_k)$. By the FFL rule, $f_i$ is the leaving arc. Therefore, $a_{j_\ell}=f_i$.
    \item By the conclusion in part \ref{it:key-2} of the lemma, the leaving column $j_\ell$ is $i$, which corresponds to the singleton $e_i$. Notice that $i\in B$ and $0\in B-O$, we have $i\in O$. Thus, $0\le u^O_i\le c_{i,i}=0$, which means $u^O_i=c_{i,i}$ and $j_\ell$ is $i$-disliked w.r.t.~$O$. By Lemma~\ref{lem:Separator-Change}, the separator of $(B',O')$ is $i'$ with some $i'>i.$
    \item We verify the $i'$-nice properties (i)-(iii) of $B'$ in order. By definition, we have $B'=B\cup\{j_t\}\setminus \{j_\ell\}$. 

    (i) Let $v\in U\setminus\{v_{n+1}\}$. Let $P_B(v_{n+1},v)$ and $P_{B'}(v_{n+1},v)$ be the $\mathcal{T}_B$-path and $\mathcal{T}_{B'}$-path from $v_{n+1}$ to $v$, respectively. Let $F\subset\mathcal{A}$ be the arcs of the cycle formed by $P_B(v_j,v_k)$ and $a_{j_t}$,~i.e.~$a\in F$ iff either $a=a_{j_t}$ or $a$ is in $P_B(v_j,v_k)$. We claim that:
    \begin{claim}\label{cl:bad-vertex}
        If $v_p$ is incident to $F$,~i.e.~$v_p$ is an endpoint of some arc in $F$, then either $v_p=v_j$ or $v_p\in W_i$. Also, $v_j\ge_{\mathcal{T}}v_p$.
    \end{claim}
    \noindent \emph{\underline{Proof of Claim~\ref{cl:bad-vertex}.}} If $v_p$ is incident to $F$, then $v_p$ is also visited by $P_B(v_j,v_k)$. By~\eqref{eq:vj-vk} and the $i$-nice property (iii) of $B$, $P_B(v_j,v_k)$ traverses from $R_i$ to $W_i$ through $f_i$, and never go back to $R_i$ as the unique bridge $f_i$ between $R_i$ and $W_i$ can be only used once. Thus, either $v_p=v_j$ or $v_p\in W_i$. If the latter happens, notice that $f_i$ is also the bridge between $R_i$ and $W_i$ on $\mathcal{T}$, the $\mathcal{T}$-directed path $P_0(v_{n+1},v_p)$ visits $v_j$, which implies $v_j\ge_{\mathcal{T}}v_p$.\hfill $ \blacksquare$    

    By the $i$-nice property (i) of $B$ and~\eqref{eq:r-vk}, $P_B(v_j,v_i)$ is $x$-augmenting and by~\eqref{eq:vj-vk}, it takes $f_i$ as a forward arc. Thus, $x_{f_i}=1$. Since $a_{j_\ell}=f_i$ and $f_i\notin \mathcal{A}_{B'}$, we have $x'_{f_i}=0$. Therefore, by Theorem~\ref{lem:Nonde-Pivot}, the cardinal pivot $B\to B'$ is non-degenerate and $P_B(v_j,v_k)$ is $x$-augmenting and $x'$-descending. Also, by the equation $x'_B-x_B=A_B^{-1}A_{j_t}(x_{j_t}-x'_{j_t})=y(-1)$, we have $x'_B=x_B-y$. By Lemma~\ref{lem:Forward-Backward-Formula}, for $s\in[m]$, $x_s\neq x'_s$ if and only if $a_s\in F$ (notice that for the non-basic column $s\notin B$, $x_s=x'_s=0$).

    
    Consider the directed graph $\mathcal{T}_{B'}=(U,\mathcal{A}_B\cup\{a_{j_t}\}\setminus\{a_{j_\ell}\})$, and let $P_{B'}(v_{n+1},v)$ be a $\mathcal{T}_{B'}$-path from $v_{n+1}$ to $v$. If $v\in R_i$, then $P_{B'}(v_{n+1},v)=P_B(v_{n+1},v)$, and they do not contain the arcs in $F$. Indeed, let $v\in R_i$ and by the $i$-nice property (iii) of $B$, if $P_B(v_{n+1},v)$ visits $W_i$, then it takes $f_i$ forwardly. However, since the destination $v\in R_i$, $P_B(v_{n+1},v)$ needs to travel back through $f_i$ backwardly, which implies that $P_B(v_{n+1},v)$ takes $f_i$ at least twice, which cannot be a path. Therefore, $P_B(v_{n+1},v)$ never visits $W_i$ and $f_i$ is not part of the path. Thus, $P_B(v_{n+1},v)$ is an $\mathcal{T}_{B'}$-path from $v_{n+1}$ to $v$. It is unique since $\mathcal{T}_{B'}$ is a tree. Thus, $P_{B'}(v_{n+1},v)=P_B(v_{n+1},v)$. By Claim~\ref{cl:bad-vertex}, $P_B(v_{n+1},v)$ does not intersect with $F$, otherwise it visits $W_i$, a contradiction. Thus, $P_{B'}(v_{n+1},v)$ also does not intersect with $F$. Notice that $P_B(v_{n+1},v)$ is $x$-augmenting, thus $P_{B'}(v_{n+1},v)$ is $x$-augmenting, which is also $x'$-augmenting since $x'_s=x_s$ for $a_s\notin F$. If $v\in W_i$, then by $i$-nice property (iii) of $B$, we have $P_B(v_{n+1},v)=P_B(v_{n+1},v_j)\oplus f_i\oplus P_B(v_i,v)$. Since $f_i\in F$, $P_B(v_{n+1},v)$ visits the arc(s) in $F$ and begins with $f_i$, as $P_B(v_{n+1},v_j)$ is fully contained in $R_i$ as we argued above (if we let $v=v_j$). Let $v_q$ be the last vertex incident to $F$ that $P_B(v_{n+1},v)$ visits, then we have 
    $$P_B(v_{n+1},v)=P_B(v_{n+1},v_j)\oplus f_i\oplus P_B(v_j,v_q)\oplus P_B(v_q,v),$$
    $$P_{B'}(v_{n+1},v)=P_{B'}(v_{n+1},v_j)\oplus(v_j,v_k)\oplus P_{B'}(v_k,v_q)\oplus P_{B'}(v_q,v),$$
    where $a_{j_t}=(v_j,v_k)\in\mathcal{A}_{B'}$. $P_{B'}(v_{n+1},v_j)$ is $x'$-augmenting since $P_B(v_{n+1},v_j)$ is $x$-augmenting, they are the same path and $x_s=x'_s$ for every $a_s$ in such path. $(v_j,v_k)=a_{j_t}$ is a path with single forward arc such that $x'_{j_t}=1$, thus is $x'$-augmenting. $P_{B'}(v_k,v_q)$ is $x'$-augmenting, since $P_B(v_k,v_q)$ is $x'$-descending, and every forward (resp.~backward) arc on $P_{B'}(v_k,v_q)$ is a backward (resp.~forward) arc on $P_B(v_k,v_q)$. $P_{B'}(v_q,v)$ is also $x'$-augmenting, since the same path $P_B(v_q,v)$ is $x$-augmenting and $x_s=x'_s$ for every $a_s$ in such path. Therefore, $P_{B'}(v_{n+1},v)$ is $x'$-augmenting.

    (ii) Notice that $\{f_i,\dots,f_n\}\subset \mathcal{A}_B$ by the $i$-nice property of $B$, then $\{f_{i+1},\dots,f_n\}\subset \mathcal{A}_{B'}$ since $\mathcal{A}_B\setminus\{f_i\}\subset \mathcal{A}_{B'}$ by $a_{j_\ell}=f_i$. Therefore, $\{f_{i'},f_{i'+1},\dots,f_n\}\subset \mathcal{A}_{B'}$ as $i'\ge i+1$.

    (iii) Let $f_q\in\{f_{i'},f_{i'+1},\dots,f_n\}$. Let $U=R,\cup W$ be the partition of vertices defined in Definition~\ref{def:nice-basis}(iii). Now, if $a\in\mathcal{A}_{B'}$ crosses between $R$ and $W$, we claim that $a=f_q$. Indeed, if we also have $a\in\mathcal{A}_{B}$, then by the $i$-nice property of $B$, we know that $a=f_q$ immediately. Otherwise, if $a\notin\mathcal{A}_{B}$, then $a=B'-B=a_{j_t}$. By $a_{j_t}=(v_j,v_k)$ and the assumption that $a$ corsses between $R$ and $W$, we have that $v_j\in R$ and $v_k\in W$. Notice that $v_j$ is the tail of $f_i=(v_j,v_i)$. If it holds $v_i\in W$, then $f_i\in\mathcal{A}_B$ also crosses between $R$ and $W$, which implies $f_i=f_q$, a contradiction with $f_q\in\{f_{i'},f_{i'+1},\dots,f_n\}$. Therefore, we have $v_i\in R$. Consider the $\mathcal{T}$-path $P_0(v_j,v_k)$ from $v_j$ to $v_k$. By the conclusion in part \ref{it:key-1} of the lemma, we know that $f_i=(v_j,v_i)$ is the first arc in $P_0(v_j,v_k)$. By the fact that $v_i\in R$ and $v_k\in W$, $f_q$ is part of $P_0(v_k,v_k)$, and appears after $f_i$. Notice that since $a_{j_t}=(v_j,v_k)\in\mathcal{A}_B$, $P_0(v_j,v_k)$ is a $\mathcal{T}$-directed path, we traverse along $P_0(v_j,v_k)$ through a sequence of vertices with decreasing order with respect to $\ge_{\mathcal{T}}$, which implies that the head of $f_i$ is greater than the head of $f_q$. Since $\mathcal{T}$ is depth-first, we have $v_i\ge_{\mathcal{T}}v_q$, and thus $i\ge q$, a contradiction with $q\ge i'>i$. 

    Therefore, $B'$ is an $i'$-nice basis.
\end{enumerate}
\end{proof}

\begin{lemma}\label{lem:Scarf-Linear-Time}
    Scarf's algorithm with FFL rule terminates in at most $n$ iterations.
\end{lemma}

\begin{proof}
We first verify that the initial bases $(B_0,O_0)$ satisfies the condition of Lemma~\ref{lem:Out-Deg-1},~i.e.~if the column $n+1=O_0-B_0$ is in $i$-th block, then $B_0$ is an $i$-nice basis.

Indeed, we have $\mathcal{A}_{B_0}=\mathcal{A}_0=\{f_1,\dots,f_n\}$ and $x_{f_j}=1$ for every $j\in[n]$. One can verify that for every $k\in[n]$, $B_0$ is an $k$-nice basis, thus in particular $B_0$ is an $i$-nice basis.
    
By Lemma~\ref{lem:Out-Deg-1}, every iteration moves the separator from $i$ to $i'$ with $i'>i$. Inductively applying this argument shows that the index of the separator strictly increases in each iteration. Thus, we have an upper bound of $n$ on the number of iterations.
\end{proof}

We now show that Scarf's algorithm will terminate in polynomial time by using Lemma~\ref{lem:Scarf-Linear-Time}. Remember that each iteration of Scarf's algorithm contains two pivots, cardinal pivot and ordinal pivot, where each single pivot can be achieved in polynomial time~\cite{scarf1967core}. Futhermore, we show the running time is indeed $O(nm)$ stated in Theorem~\ref{thm:main-arb-poly}.
\begin{proof}[Proof of Theorem~\ref{thm:main-arb-poly}.] 
    To embed the input into Scarf's algorithm with tuple $(A,b,C)$, we can construct the matrices $A,C$ in time $O(m)$. 

    The cardinal pivot with the implementation present in Algorithm~\ref{alg:cp} involves finding a path between two vertices in a given tree, and checking whether such path is $x$-augmenting or not. Since the length of such path is bounded by $n$, the running time is $O(n)$, and thus by Lemma~\ref{lem:Scarf-Linear-Time}, the total running time on cardinal pivots is $O(n^2)$.
    

    The ordinal pivot shown in Algorithm~\ref{alg:op} needs to find set $K$ where the ``candidate'' entering column belongs to, which needs at most $mn$ times of comparisons. However, according to Lemma~\ref{lem:Separator-Change} and Lemma~\ref{lem:Out-Deg-1}, for every iteration we change the separator of the Scarf pair, and such ordinal pivots are generally easier. In fact, we can rewrite ordinal pivot (Algorithm~\ref{alg:op}) in our problem as Algorithm~\ref{alg:op-arb}. Notice that, as $j^*$ is the maximizer of $c_{0,k}$ for $k\in K$, which is equivalent to say $j^*$ is the smallest index in $K$, since we have $c_{0,0}>c_{0,1}>\dots>c_{0,m}$. Therefore, we do not need to find the whole set $K$, instead we start from the block $S_{i+1}$ and from left to right search if there exists a column $j>j_r$ such that for all $\bar{\imath}\neq 0$, $c_{\bar{\imath},j}>u_{\bar{\imath}}^{O-i}$. Thus, in an iteration when we move the separator from $i$ to $i'$, we only search the columns between the block $S_{i+1}$ and $S_{i'}$, which means at most $m$ columns are searched among all iterations. Since for each searching round we need to compare the $n$ entries between two columns, the total running time for all ordinal pivots is $O(nm)$. 

    Therefore, the total running time of Scarf's algorithm with our implementation is bounded by $O(nm)$.
\end{proof}

\begin{algorithm}
\caption{Ordinal Pivot with Separator Change}\label{alg:op-arb}
\begin{algorithmic}
\State{Let $(O,i)$ be the current ordinal basis with separator $i$, associated with utility vector $u^O$.}
\State{$j_\ell\gets i$}\Comment{Leaving column is set to correspond to the arc $f_i$ (Lemma~\ref{lem:Out-Deg-1}(ii))}
\State{$i_\ell\gets i$}\Comment{The arc $f_i$ is $i$-disliked w.r.t.~$O$}
\State{$j_r\gets  \max O$}\Comment{Lemma~\ref{lem:Separator-Change}(i)}
\State{$i_r\gets 0$}\Comment{Lemma~\ref{lem:Separator-Change}(ii)}
\If{$\{j>j_r:c_{\bar{\imath},j}>u_{\bar{\imath}}^{O-i},\forall \bar{\imath}\neq 0\}=\emptyset$}
\State{$O'=O\cup\{0\}\setminus\{j_\ell\}$\textrm{ is a dominating basis for $(A,b,C)$.}}
\Else
\State{$j^*\gets\min\{j>j_r:c_{\bar{\imath},j}>u_{\bar{\imath}}^{O-i},\forall \bar{\imath}\neq 0\}$}\Comment{Decreasing order $c_{0,0}>c_{0,1}>\dots>c_{0,m}$}
\State{$O\gets O\cup\{j^*\}\setminus\{j_\ell\}$}
\State{Let $j^*$ belong to the block $S_{i'}$ in matrix $C$.}
\State{$i\gets i'$}
\EndIf
\end{algorithmic}
\end{algorithm}

\section{Non-integrality of the Fractional Stable Matching Polytope}\label{sec:negative}

A classical approach to stable matching problem in graphs is to describe the stable matching polytope~\cite{teo1998geometry}, that is, find an polyhedral description of the convex hull of all characteristic vectors of stable matchings. This approach has been successful for the classical marriage model~\cite{rothblum1992characterization,vate1989linear}, as well as some of its generalizations~\cite{faenza2023affinely,fleiner2003stable}. We construct a hypergraphic preference system $I=(H=(V,E),\succ)$ where $H$ is an interval hypergraph, on which the above classical approach does not succeed. This result provides evidence that the stable matching problem in arborescence hypergraphic preference system is technically challenging compared to the classical stable matching problem in bipartite graphs. 

We first define 
\begin{equation}
    P(I)=conv\left(\left\{x\in \{0,1\}^E:x\textrm{ is a stable matching} 
    \right\}\right),
\end{equation}
and
\begin{equation}
    Q(I)=\left\{x\in \mathbb{R}^E\,\vline
    \begin{array}{cc}
      x(\delta(i)) \le 1,  &  \forall i\in V, \\
      x(e^\succeq) \ge 1,   & \forall e\in E, \\
      0\le x_e \le 1, & \forall e\in E.
    \end{array}
    \right\},
\end{equation}
where $\delta(i)=\{e\in E:i\in e\}$ and $e^\succeq=\{e'\in  E:\exists i\in e, e'\succ_i e\}\cup\{e\}$. We call $P(I)$ the \emph{stable matching polytope} of $I$ and $Q(I)$ the \emph{fractional stable matching polytope} of $I$. When $H$ is a bipartite graph, it is known that $P(I)=Q(I)$~(see,~e.g.,~\cite{teo1998geometry}). In contrast, we show that this equality fails to hold for hypergraph preference systems where the hypergraph is an interval hypergraph. 
\begin{theorem}\label{thm:neg}
    There is a hypergraphic preference system $I=(H=(V,E),\succ)$ where $H$ is an interval hypergraph, such that the fractional stable matching polytope $Q(I)$ is not integral, thus $P(I)\neq Q(I)$.
\end{theorem}
We give the following example that shows Theorem~\ref{thm:neg}.

\begin{example}\label{ex:Non-Integral-Ex}

\begin{figure}[h!]
    \centering
    \includegraphics[scale=1.4]{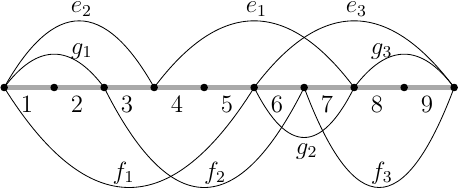}
    \caption{The underlying hypergraph in Example~\ref{ex:Non-Integral-Ex}. A node belongs to an edge iff the latter covers the former in this figure.}
    \label{fig:Non-Integral-Ex}
\end{figure}
    Let $I=(H=(V,E),\succ)$ where $V=\{1,2,\dots,9\}$, $E=\{e_1,e_2,e_3,f_1,f_2,f_3,g_1,g_2,g_3\}$ are shown in Figure~\ref{fig:Non-Integral-Ex}. The preference list is as follows:
    \begin{equation*}
        \begin{split}
            & 1: e_2 \succ_1 g_1 \succ_1 f_1 ,\\
            & 2: f_1 \succ_2 e_2 \succ_2 g_1 ,\\
            & 3: f_2 \succ_3 f_1 \succ_3 e_2 ,\\
            & 4: e_1 \succ_4 f_2 \succ_4 f_1 ,\\
            & 5: f_1 \succ_5 e_1 \succ_5 f_2 ,\\
            & 6: e_1 \succ_6 f_2 \succ_6 e_3 \succ_6 g_2 ,\\
            & 7: f_3 \succ_7 e_1 \succ_7 e_3 \succ_7 g_2 ,\\
            & 8: e_3 \succ_8 g_3 \succ_8 f_3 ,\\
            & 9: f_3 \succ_9 e_3 \succ_9 g_3 . 
        \end{split}
    \end{equation*}
    Consider the solution $x$ where $x(e_1)=x(e_2)=x(e_3)=0$, $x(f_1)=x(f_2)=x(f_3)=x(g_1)=x(g_2)=x(g_3)=\frac{1}{2}$. We show that $x$ is an extreme point of $Q(I)$.

    First, observe that $\{g_1,f_2,f_3\}$ and $\{f_1,g_2,g_3\}$ are two partitions of the interval, thus every $i\in I$ yields $x(\delta(i))=\frac{1}{2}+\frac{1}{2}=1$. One can check that the stability constraints are satisfied:

    \begin{equation*}
        \begin{split}
            & x(e_1^\succeq)=x(e_1)+x(f_1)+x(f_3)=1 ,\\
            & x(e_2^\succeq)=x(e_2)+x(f_1)+x(f_2)=1 ,\\
            & x(e_3^\succeq)=x(e_3)+x(e_1)+x(f_2)+x(f_3)=1 ,\\
            & x(f_1^\succeq)=x(f_1)+x(e_2)+x(g_1)+x(f_2)+x(e_1)=\frac{3}{2}>1 ,\\
            & x(f_2^\succeq)=x(f_2)+x(e_1)+x(f_1)+x(e_1)=1 ,\\
            & x(f_3^\succeq)=x(f_3)+x(e_3)+x(g_3)=1 ,\\
            & x(g_1^\succeq)=x(g_1)+x(e_2)+x(f_1)+x(e_2)=1 ,\\
            & x(g_2^\succeq)=x(g_2)+x(e_1)+x(f_2)+x(e_3)+x(f_3)=\frac{3}{2}>1 ,\\
            & x(g_3^\succeq)=x(g_3)+x(e_3)+x(f_3)=1 . 
        \end{split}
    \end{equation*}

    Now, pick the tight constraints corresponding to $e_1^\succeq, e_2^\succeq, e_3^\succeq$, $\delta(1),\delta(6),\delta(9)$ and restrict them to the positive iariables $x(f_i),x(g_i)$ with $i=1,2,3$, we find a linear system

\begin{displaymath}
\begin{pNiceMatrix}[first-row,first-col]
     & f_1 & f_2 & f_3 & g_1 & g_2 & g_3 & e_1 & e_2 & e_3 \\
    e_1^\succeq & 1 &   & 1 &   &   &   & 1 &   &   \\
    e_2^\succeq & 1 & 1 &   &   &   &   &   & 1 &   \\
    e_3^\succeq &   & 1 & 1 &   &   &   & 1 &   & 1 \\
    \delta(1)   & 1 &   &   & 1 &   &   &   & 1 &   \\
    \delta(6)   &   & 1 &   &   & 1 &   & 1 &   & 1 \\
    \delta(9)   &   &   & 1 &   &   & 1 &   &   & 1 \\
    \hdashline
    x(e_1)      &   &   &   &   &   &   & 1 &   &   \\
    x(e_2)      &   &   &   &   &   &   &   & 1 &   \\
    x(e_3)      &   &   &   &   &   &   &   &   & 1  
\end{pNiceMatrix}
x=
\begin{pNiceMatrix}[first-row]
     RHS \\
     1 \\
     1 \\
     1 \\
     1 \\
     1 \\
     1 \\
     0 \\
     0 \\
     0 
\end{pNiceMatrix},
\end{displaymath}
where the left-hand-side coefficient matrix has full rank. In fact, it is sufficient to check the upper-left $6\times 6$ matrix has full rank. Thus, we find $9$ linearly independent tight constraints that uniquely determine $x$.

Therefore, $x$ is indeed an extreme point of $Q(I)$, however it is not integral.
\end{example}

\section{Conclusions and Future Work}

We showed that Scarf's algorithm converges in polynomial time and returns an integral stable matching on arborescence hypergraphic preference systems. Our result is the first proof of polynomial-time convergence of Scarf's algorithm on hypergraphic stable matching problems. We note that some of our results hold for hypergraphs that are more general than arborescence hypergraphs. We mention three directions for future research: Firstly, it would be interesting to generalize our approach to show  polynomial-time convergence of Scarf's algorithm for more general classes of hypergraphs, such as network hypergraphs. Secondly, it would be insightful to find an interpretation as a purely combinatorial algorithm of our implementation of Scarf's algorithm for arborescence hypergraphs, if any such interpretation exists. Thirdly, is it possible to verify whether a given network hypergraph is an arborescence hypergraph in polynomial-time? More generally, is it possible to construct a principal tree associated with a given network hypergraph with the least number of sources? We leave the above as open questions.

\medskip

\noindent {\bf Acknowledgments.} Yuri Faenza and Chengyue He are supported by the NSF Grant 2046146 and by a Meta Research Award. Karthekeyan Chandrasekaran is supported in part by the NSF grant CCF-2402667. Part of this work was completed while the first three authors were attending the semester program on \emph{Discrete Optimization: Mathematics, Algorithms, and Computation} at ICERM, Brown University, USA.

\bibliographystyle{abbrv}
\bibliography{ref}

\begin{thebibliography}{10}

\bibitem{aharoni2003lemma}
R.~Aharoni and T.~Fleiner.
\newblock On a lemma of {S}carf.
\newblock {\em Journal of Combinatorial Theory, Series B}, 87(1):72--80, 2003.

\bibitem{aharoni1995fractional}
R.~Aharoni and R.~Holzman.
\newblock Fractional kernels in digraphs.
\newblock {\em J. Combinatorial Theory}, 73(1):1--6, 1998.

\bibitem{baiou2000stable}
M.~Ba{\"\i}ou and M.~Balinski.
\newblock The stable admissions polytope.
\newblock {\em Mathematical programming}, 87:427--439, 2000.

\bibitem{baiou2002stable}
M.~Ba{\"\i}ou and M.~Balinski.
\newblock The stable allocation (or ordinal transportation) problem.
\newblock {\em Mathematics of Operations Research}, 27(3):485--503, 2002.

\bibitem{biro2016fractional}
P.~Bir{\'o} and T.~Fleiner.
\newblock Fractional solutions for capacitated ntu-games, with applications to
  stable matchings.
\newblock {\em Discrete Optimization}, 22:241--254, 2016.

\bibitem{cormen2022introduction}
T.~H. Cormen, C.~E. Leiserson, R.~L. Rivest, and C.~Stein.
\newblock {\em Introduction to algorithms}.
\newblock MIT press, 2022.

\bibitem{csaji2022complexity}
G.~Cs{\'a}ji.
\newblock On the complexity of stable hypergraph matching, stable
  multicommodity flow and related problems.
\newblock {\em Theoretical Computer Science}, 931:1--16, 2022.

\bibitem{dworczak2016deferred}
P.~Dworczak.
\newblock Deferred acceptance with compensation chains.
\newblock In {\em Proceedings of the 2016 ACM Conference on Economics and
  Computation}, pages 65--66, 2016.

\bibitem{faenza2023scarf}
Y.~Faenza, C.~He, and J.~Sethuraman.
\newblock Scarf's algorithm and stable marriages.
\newblock {\em arXiv preprint arXiv:2303.00791}, 2023.

\bibitem{faenza2023affinely}
Y.~Faenza and X.~Zhang.
\newblock Affinely representable lattices, stable matchings, and choice
  functions.
\newblock {\em Mathematical Programming}, 197(2):721--760, 2023.

\bibitem{fleiner2003stable}
T.~Fleiner.
\newblock On the stable b-matching polytope.
\newblock {\em Mathematical Social Sciences}, 46(2):149--158, 2003.

\bibitem{gale1962college}
D.~Gale and L.~S. Shapley.
\newblock College admissions and the stability of marriage.
\newblock {\em The American Mathematical Monthly}, 69(1):9--15, 1962.

\bibitem{haxell2008fractional}
P.~E. Haxell and G.~T. Wilfong.
\newblock A fractional model of the border gateway protocol (bgp).
\newblock In {\em Proceedings of the nineteenth annual ACM-SIAM symposium on
  Discrete algorithms}, pages 193--199. Citeseer, 2008.

\bibitem{ishizuka2018complexity}
T.~Ishizuka and N.~Kamiyama.
\newblock On the complexity of stable fractional hypergraph matching.
\newblock In {\em 29th International Symposium on Algorithms and Computation
  (ISAAC 2018)}. Schloss-Dagstuhl-Leibniz Zentrum f{\"u}r Informatik, 2018.

\bibitem{kintali2013reducibility}
S.~Kintali, L.~J. Poplawski, R.~Rajaraman, R.~Sundaram, and S.-H. Teng.
\newblock Reducibility among fractional stability problems.
\newblock {\em SIAM Journal on Computing}, 42(6):2063--2113, 2013.

\bibitem{knuth1976marriages}
D.~E. Knuth.
\newblock Marriages stables.
\newblock {\em Technical report}, 1976.

\bibitem{ng1991three}
C.~Ng and D.~S. Hirschberg.
\newblock Three-dimensional stabl matching problems.
\newblock {\em SIAM Journal on Discrete Mathematics}, 4(2):245--252, 1991.

\bibitem{PIM20}
A.~Pass-Lanneau, A.~Igarashi, and F.~Meunier.
\newblock Perfect graphs with polynomially computable kernels.
\newblock {\em Discrete Applied Mathematics}, 272:69--74, 2020.

\bibitem{roth1993stable}
A.~E. Roth, U.~G. Rothblum, and J.~H. Vande~Vate.
\newblock Stable matchings, optimal assignments, and linear programming.
\newblock {\em Mathematics of operations research}, 18(4):803--828, 1993.

\bibitem{rothblum1992characterization}
U.~G. Rothblum.
\newblock Characterization of stable matchings as extreme points of a polytope.
\newblock {\em Mathematical Programming}, 54:57--67, 1992.

\bibitem{scarf1967core}
H.~E. Scarf.
\newblock The core of an n person game.
\newblock {\em Econometrica: Journal of the Econometric Society}, pages 50--69,
  1967.

\bibitem{schrijver1998theory}
A.~Schrijver.
\newblock {\em Theory of linear and integer programming}.
\newblock John Wiley \& Sons, 1998.

\bibitem{teo1998geometry}
C.-P. Teo and J.~Sethuraman.
\newblock The geometry of fractional stable matchings and its applications.
\newblock {\em Mathematics of Operations Research}, 23(4):874--891, 1998.

\bibitem{vate1989linear}
J.~H.~V. Vate.
\newblock Linear programming brings marital bliss.
\newblock {\em Operations Research Letters}, 8(3):147--153, 1989.

\end{thebibliography}

\end{document}